\documentclass[12 pt]{article}
\usepackage{lmodern} 
\usepackage[T1]{fontenc}
\usepackage{textcomp}
\usepackage{multirow}
\usepackage{tabularx} 
\newcolumntype{C}{>{\centering\arraybackslash}X} 
\usepackage[scr=boondoxo,frak=boondox,bb=boondox]{mathalfa}
\usepackage{amsmath}             
\allowdisplaybreaks
\usepackage{amssymb}
\usepackage{bm}
\usepackage{amsthm}
 
\usepackage{ragged2e}
\usepackage{breqn}
\usepackage{graphicx}
\usepackage{epstopdf}
\usepackage{enumitem}
\def\namedlabel#1#2{\begingroup
    #2%
    \def\@currentlabel{#2}%
    \phantomsection\label{#1}\endgroup
}
\usepackage[title]{appendix}
\usepackage{lipsum}
\usepackage{authblk}
\usepackage[top=2cm, bottom=2cm, left=2cm, right=2cm]{geometry}
\usepackage{fancyhdr}

\usepackage{appendix}
\usepackage{hyperref}

\usepackage{enumitem}
\usepackage{float}
\usepackage{verbatim}
\usepackage{booktabs}
\renewenvironment{abstract}{
\hfill\begin{minipage}{0.95\textwidth}
\rule{\textwidth}{1pt}}
{\par\noindent\rule{\textwidth}{1pt}\end{minipage}}

\makeatletter
\renewcommand\@maketitle{
\hfill
\begin{minipage}{0.95\textwidth}
\vskip 2em
\let\footnote\thanks 
{\LARGE \@title \par }
\vskip 1.5em
{\large \@author \par}
\end{minipage}
\vskip 1em \par
}
\makeatother

\newtheorem{theorem}{Theorem}
\newtheorem{lemma}{Lemma}
\newtheorem{assumption}{Assumption}
\title{\textbf{\huge{Chirp-like model and its parameters estimation}}}
\author[$\dagger$]{Rhythm Grover}
\author[$\dagger$,$\ddagger$]{Debasis Kundu}
\author[$\dagger$]{Amit Mitra}
\affil[$\dagger$]{Department of Mathematics, Indian Institute of Technology Kanpur \\
Kanpur - 208016, India}
\affil[$\ddagger$]{Corresponding author. Email: kundu@iitk.ac.in}
\date{}
\begin{document}
\maketitle

\begin{abstract}
Abstract: We propose a chirp-like signal model as an alternative to a chirp model and a generalisation of the sinusoidal model, which is a fundamental model in the statistical signal processing literature. It is observed that the proposed model can be arbitrarily close to the chirp model. The propounded model is similar to a chirp model in the sense that here also the frequency changes linearly with time. However, the parameter estimation of a chirp-like model is simpler compared to a chirp model. In this paper, we consider the least squares and the sequential least squares estimation procedures and study the asymptotic properties of these proposed estimators. These asymptotic results are corroborated through simulation studies and analysis of four speech signal data sets have been performed to see the effectiveness of the proposed model, and the results are quite encouraging.\\

\end{abstract}
\section{Introduction}
One of the most extensively used models in statistical signal processing literature is the \textit{sinusoidal model}, which can be used to model real life phenomena that are periodic in nature. For instance, ECG signals, speech and audio signals, low and high tides of the ocean, daily temperature of a city and many more. Mathematical expression for the sinusoidal model is:
$$y(t) = \sum_{j=1}^{p} \{A_j^0 \cos(\alpha_j^0 t) + B_j^0 \sin(\alpha_j^0 t)\} + X(t).$$
Here, $A_j^0$s, $B_j^0$s are the amplitudes, $\alpha_j^0$s are the frequencies and $X(t)$ is the noise component. Extensive work has been carried out on this and some related models (see for example the recent monograph by Kundu and Nandi \cite{2012} in this topic). \\\par
Another prevalent model in signal processing, is the \textit{chirp model}. This model is mathematically expressed as follows:
\begin{equation}\label{chirp_model}
y(t) = \sum_{j=1}^{p} \{A_j^0 \cos(\alpha_j^0 t + \beta_j^0 t^2) + B_j^0 \sin(\alpha_j^0 t + \beta_j^0 t^2)\} + X(t),
\end{equation}
where $\beta_j^0$s are known as the frequency rates and $A_j^0$s, $B_j^0$s, $\alpha_j^0$s  and $X(t)$ are same as before. It is evident from the model equations, that the chirp model is a natural extension of the sinusoidal model where the frequency instead of being constant, changes linearly with time. This signal is observed in many natural as well as fabricated systems and in many fields of science and engineering, like sonar and radar systems, in echolocation, audio and speech signals, in biomedical systems like EEG, EMG, and communications etc. One major issue related to chirp model is the efficient estimation of the unknown parameters. Some of the references in this area are  Abatzoglou \cite{1986}, Djuric and Kay \cite{1990_1}, Peleg and Porat \cite{1991_1}, Ikram \textit{et al}. \cite{1997}, Saha and Kay \cite{2002}, Nandi and Kundu \cite{2004}, Kundu and Nandi \cite{2008_1}, Lahiri \textit{et al}. \cite{2014},  \cite{2015}, Mazumder \cite{2016}, Grover \textit{et al}. \cite{2018} and the references cited therein. One of the widely used methods for estimation of parameters of both linear and nonlinear models, is the least squares estimation method. However, finding the least squares estimators (LSEs) for a chirp model is computationally challenging as the least squares surface is highly nonlinear. For details, see Kundu and Nandi \cite{2008_1}, Lahiri\textit{ et al}. \cite{2015} and Grover \textit{et al}. \cite{2018}. \\ \par
In this paper, we propose a new model, a \textit{chirp-\textit{like} model}, that can be expressed mathematically as follows:
\begin{equation}\label{multiple_comp_model}
y(t)  = \sum_{j=1}^{p} \{A_j^0 \cos(\alpha_j^0 t) + B_j^0 \sin(\alpha_j^0 t)\} + \sum_{k=1}^{q} \{C_k^0 \cos(\beta_k^0 t^2) + D_k^0 \sin(\beta_k^0 t^2)\} + X(t),
\end{equation}
where $A_j^0$s, $B_j^0$s, $C_k^0$s, $D_k^0$s are the amplitudes, $\alpha_k^0$s and $\beta_k^0$s are the frequencies and frequency rates, respectively. Note that if $C_k^0 = D_k^0 = 0;\ k = 1, \cdots, q,$ then this model reduces to a sinusoidal model as defined above. It is observed that the model \eqref{multiple_comp_model} behaves very similar to the model \eqref{chirp_model}. \\ \par
Recently, Grover \textit{et al}. \cite{2018}, analysed a sound vowel data "AAA" using a multiple component chirp model. Here, we re-analyse the same data set using the proposed chirp-like model. In the following figure, we plot the two "\textit{best}" fitted models. It is clear that they are well-matched, and both of them fit the original data very well. 
\begin{figure}[H]
\centering
\includegraphics[scale=0.3]{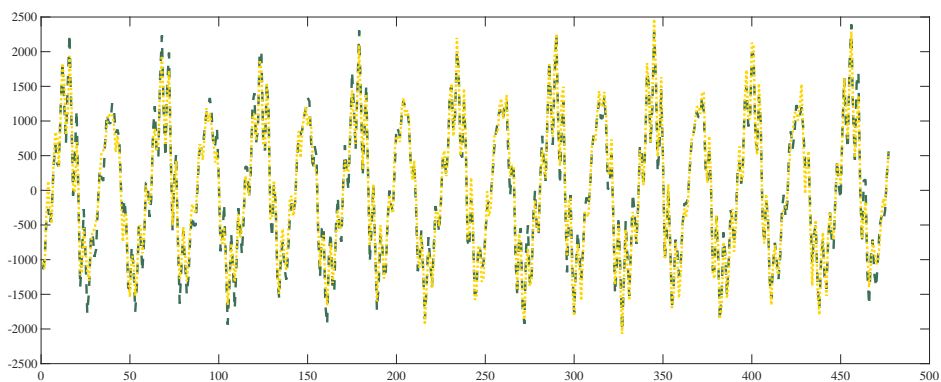}
\caption{Fitted chirp signal (green dashed line) and fitted chirp-like signal (yellow dotted line) to the "AAA" sound data}
\label{fig_fittings}
\end{figure} 

\justify 
It is apparent that the data which can be analysed using the chirp model can also be modelled using the proposed chirp-like model. Moreover, it is observed in this paper that computation of the LSEs of the chirp-like model is much simpler compared to the chirp model. The more detailed explanation will be provided in Section \ref{data_analysis}.\\ \par  

For the estimation of the parameters of the chirp-like model, we first consider the usual LSEs and study their asymptotic properties. Since it is observed that the computation of the LSEs is a numerically challenging problem, we propose a sequential procedure which reduces the computational burden significantly. This procedure follows along the same lines as the one proposed by Prasad \textit{et al.} \cite{2008_2} for the multiple component sinusoidal model and Lahiri \textit{et al.} \cite{2015} for the chirp model. We obtain the consistency and the asymptotic distribution of the sequential estimators as well. It is observed that both these estimators are strongly consistent and they have the same asymptotic distribution.\\ \par

The rest of the paper is organised as follows. In the next section, we define the one component chirp-like model and propose the least squares estimation and the sequential estimation of the parameters of this model and study the asymptotic properties of the obtained estimators. In Section \ref{multiple_component_theory}, we study the asymptotic properties of a more generalised model as defined in  \eqref{multiple_comp_model}. In Section \ref{simulation_results}, we perform some simulations to validate the asymptotic results and in Section \ref{data_analysis}, we analyse four speech signal data sets to see how the proposed model performs in practice. We conclude the paper in Section \ref{conclusion}. The tables, preliminary results and all the proofs are provided in the appendices.
      
\section{One Component Chirp-like Model}\label{one_component_theory}
In this section, we consider a one component chirp-like model, expressed mathematically as follows:
\begin{equation}\label{one_comp_model_eq}
y(t) = A^0 \cos(\alpha^0 t) + B^0 \sin(\alpha^0 t) + C^0 \cos(\beta^0 t^2) + D^0 \sin(\beta^0 t^2) + X(t).
\end{equation}
Our purpose is to estimate the unknown parameters of the model, the amplitudes  $A^0$, $B^0$, $C^0$, $D^0$, the frequency  $\alpha^0$ and the frequency rate $\beta^0$ under the following assumption on the noise component:

\begin{assumption}
\label{assump:1} 
Let $Z$ be the set of integers. $\{X(t)\}$ is a stationary linear process of the form:
\begin{equation}\label{error_assumption}
X(t) = \sum_{j = -\infty}^{\infty} a(j)e(t-j),
\end{equation}
where $\{e(t); t \in Z\}$ is a sequence of  independently and identically distributed (i.i.d.) random variables with $E(e(t)) = 0$, $V(e(t)) = \sigma^2$, and $a(j)$s are real constants such that 
\begin{equation}\label{condition_on_constants}
\sum\limits_{j= - \infty}^{\infty}|a(j)| < \infty.
\end{equation} 
\end{assumption}
\justify
This is a standard assumption for a stationary linear process. Any finite dimensional stationary MA, AR or ARMA process can be represented as~\eqref{error_assumption} when the coefficients $a(j)$s satisfy condition~\eqref{condition_on_constants} and hence this covers a large class of stationary random variables. \\
\justify
We will use the following notations for further development: $\boldsymbol{\theta}$ = $(A, B, \alpha, C, D, \beta)$, the parameter vector and $\boldsymbol{\theta}^0$ = $(A^0, B^0, \alpha^0, C^0, D^0, \beta^0)$ as the true parameter vector and $\boldsymbol{\Theta}$ = $[-M, M] \times [-M, M]  \times [0,\pi] \times [-M, M] \times [-M, M] \times [0,\pi]$, where $M$ is a positive real number. Also we make the following assumption on the unknown parameters:
\begin{assumption}\label{assump:2}
The true parameter vector $\boldsymbol{\theta}^0$ is an interior point of the parametric space $\boldsymbol{\Theta}$, and ${A^0}^2 + {B^0}^2 + {C^0}^2 + {D^0}^2 > 0$. 
\end{assumption}
\justify
Under these assumptions, we discuss two estimation procedures, the least squares estimation method and the sequential least squares estimation method. We then study the asymptotic properties of the estimators obtained using these methods.
  
\subsection{Least Squares Estimators}\label{one_component_LSEs_theory}
The usual LSEs of the unknown parameters of model \eqref{one_comp_model_eq} can be obtained by minimising the error sum of squares:
$$Q(\boldsymbol{\theta}) = \sum_{t=1}^{n}\bigg(y(t) - A\cos(\alpha t) - B \sin(\alpha t) - C\cos(\beta t^2) - D\sin(\beta t^2)\bigg)^2,$$
with respect to $A$, $B$, $\alpha$, $C$, $D$ and $\beta$ simultaneously. In matrix notation,
\begin{equation}\label{ess_matrix}
Q(\boldsymbol{\theta}) = (\textbf{Y} - \textbf{Z}(\alpha, \beta)\boldsymbol{\mu})^{T}(\textbf{Y} - \textbf{Z}(\alpha, \beta)\boldsymbol{\mu}).
\end{equation}
\justify
Here $\textbf{Y}_{n \times 1} = \begin{pmatrix}
y(1) \\ \vdots \\ y(n)\end{pmatrix}$, $\boldsymbol{\mu}_{4 \times 1}  = \begin{pmatrix}
A \\ B \\ C \\ D
\end{pmatrix}$ and $\textbf{Z}(\alpha, \beta)_{n \times 4} = \begin{pmatrix}
\cos(\alpha) & \sin(\alpha) & \cos(\beta) & \sin(\beta) \\
\vdots & \vdots & \vdots & \vdots \\
\cos(n \alpha) & \sin(n \alpha) & \cos(n^2 \beta) & \sin(n^2 \beta)
\end{pmatrix}.$ \\
Since $\boldsymbol{\mu}$ is a vector of linear parameters, by separable linear regression technique of Richards \cite{1961}, we have:
\begin{equation}\label{linear_par_estimates}
\hat{\boldsymbol{\mu}}(\alpha, \beta) = [\textbf{Z}(\alpha, \beta)^{T}\textbf{Z}(\alpha, \beta)]^{-1}\textbf{Z}(\alpha, \beta)^{T} \textbf{Y}.
\end{equation}
Using  \eqref{linear_par_estimates} in \eqref{ess_matrix}, we obtain:
\begin{align*}
R(\alpha, \beta) & = Q(\hat{A}(\alpha, \beta), \hat{B}(\alpha, \beta), \alpha, \hat{C}(\alpha, \beta), \hat{D}(\alpha, \beta), \beta)\\  
& = \textbf{Y}^{T}(\textbf{I} - \textbf{Z}(\alpha, \beta)[\textbf{Z}(\alpha, \beta)^{T}\textbf{Z}(\alpha, \beta)]^{-1}\textbf{Z}(\alpha, \beta)^{T})\textbf{Y}.
\end{align*}
To obtain $\hat{\alpha}$ and $\hat{\beta}$, the LSEs of  $\alpha^0$ and $\beta^0$ respectively, we minimise $R(\alpha, \beta)$ with respect to $\alpha$ and $\beta$ simultaneously. Once we obtain $\hat{\alpha}$ and $\hat{\beta}$, by substituting them in \eqref{linear_par_estimates}, we obtain the LSEs of the linear parameters. \par
\justify
The following results provide the consistency and asymptotic normality properties of the LSEs.
\begin{theorem}\label{consistency_LSEs_one_component}
Under assumptions \ref{assump:1} and \ref{assump:2}, $\hat{\boldsymbol{\theta}} = (\hat{A}, \hat{B}, \hat{\alpha}, \hat{C}, \hat{D}, \hat{\beta})$ is a strongly consistent estimator of $\boldsymbol{\theta}^0$, that is, \\
$$\hat{\boldsymbol{\theta}} \xrightarrow{a.s.} \boldsymbol{\theta}^0 \textmd{ as } n \rightarrow \infty.$$ 
\end{theorem}
\begin{proof}
See Appendix ~\ref{appendix:C1}.\\
\end{proof}
\begin{theorem}\label{asymptotic_dist_LSEs_one_component}
Under assumptions \ref{assump:1} and \ref{assump:2}, 
$$(\hat{\boldsymbol{\theta}} - \boldsymbol{\theta}^0) \textbf{D}^{-1} \xrightarrow{d} \mathcal{N}(0, c \sigma^2 \boldsymbol{\Sigma}^{-1}(\boldsymbol{\theta}^0)),$$
where $\textbf{D} = diag(\frac{1}{\sqrt{n}}, \frac{1}{\sqrt{n}}, \frac{1}{n\sqrt{n}}, \frac{1}{\sqrt{n}}, \frac{1}{\sqrt{n}}, \frac{1}{n^2\sqrt{n}})$, $c = \sum\limits_{j=-\infty}^{\infty} a(j)^2$ and\\
$$\boldsymbol{\Sigma}^{-1}(\boldsymbol{\theta}^0) = \begin{pmatrix}
\frac{2({A^0}^2 + 4 {B^0}^2)}{{A^0}^2 + {B^0}^2} & \frac{-6A^0B^0}{{A^0}^2 + {B^0}^2} & \frac{-12B^0}{{A^0}^2 + {B^0}^2} & 0 & 0 & 0 \\
\frac{-6A^0B^0}{{A^0}^2 + {B^0}^2} & \frac{2(4 {A^0}^2 + {B^0}^2)}{{A^0}^2 + {B^0}^2} & \frac{12A^0}{{A^0}^2 + {B^0}^2} & 0 & 0 & 0 \\
\frac{-6A^0B^0}{{A^0}^2 + {B^0}^2} & \frac{2(4 {A^0}^2 + {B^0}^2)}{{A^0}^2 + {B^0}^2} & \frac{12A^0}{{A^0}^2 + {B^0}^2} & 0 & 0 & 0 \\
0 & 0 & 0 &  \frac{4{C^0}^2 + 9 {D^0}^2}{2({C^0}^2 + {D^0}^2)} & \frac{-5 C^0 D^0}{2({C^0}^2 + {D^0}^2)} & \frac{-15D^0}{2({C^0}^2 + {D^0}^2)}\\
0 & 0 & 0 &  \frac{-5 C^0 D^0}{2({C^0}^2 + {D^0}^2)} & \frac{9 {C^0}^2 + 4 {D^0}^2}{2({C^0}^2 + {D^0}^2)} & \frac{15 C^0}{2({C^0}^2 + {D^0}^2)}\\
0 & 0 & 0 &  \frac{-15D^0}{2({C^0}^2 + {D^0}^2)} & \frac{15 C^0}{2({C^0}^2 + {D^0}^2)} & \frac{45}{2({C^0}^2 + {D^0}^2)}\\
\end{pmatrix}.$$
\end{theorem}
\begin{proof}
See Appendix ~\ref{appendix:C1}.\\
\end{proof}
\justify
Note that to estimate the frequency and frequency rate parameters, we need to solve a 2D nonlinear optimisation problem. Even for a particular case of this model, when $C^0 = D^0 = 0$, it has been observed that the least squares surface is highly nonlinear and has several local minima near the true parameter value (for details, see Rice and Rosenblatt \cite{1988}). Therefore, it is evident that computation of the LSEs is a numerically challenging problem for the proposed model as well. 
\subsection{Sequential Least Squares Estimators}\label{one_component_seq_LSEs_theory}
In order to overcome the computational difficulty of finding the LSEs without compromising on the efficiency of the estimates, we propose a sequential procedure to find the estimates of the unknown parameters of model \eqref{one_comp_model_eq}. In this section, we present the algorithm to obtain the sequential estimators and study the asymptotic properties of these estimators.\\ \\
Note that the matrix $\textbf{Z}(\alpha, \beta)$ can be partitioned into two $n \times 2$ blocks as follows:\\
$$\textbf{Z}(\alpha, \beta) = \begin{pmatrix}\begin{array}{c|c}
\textbf{Z}^{(1)}(\alpha) & \textbf{Z}^{(2)}(\beta)
\end{array}
\end{pmatrix}.$$
Here $\textbf{Z}^{(1)}(\alpha)_{n \times 2} = \begin{pmatrix}
\cos(\alpha) & \sin(\alpha) \\
\vdots & \vdots \\
\cos(n \alpha) & \sin(n \alpha)
\end{pmatrix}$ and $\textbf{Z}^{(2)}(\beta)_{n \times 2} = \begin{pmatrix}
\cos(\beta) & \sin(\beta) \\
\vdots & \vdots \\
\cos(n^2 \beta) & \sin(n^2 \beta)
\end{pmatrix}.$ Similarly, $\boldsymbol{\mu} = \begin{pmatrix}
\begin{array}{c} \boldsymbol{\mu}^{(1)} \\ \hline \boldsymbol{\mu}^{(2)} \end{array} \end{pmatrix}$, where $\boldsymbol{\mu}^{(1)}_{2 \times 1} = \begin{pmatrix}
A \\ B
\end{pmatrix}$ and  $\boldsymbol{\mu}^{(2)}_{2 \times 1} = \begin{pmatrix}
C \\ D
\end{pmatrix}.$ Also, the parameter vector, $\boldsymbol{\theta} = \begin{pmatrix}
\begin{array}{c} \boldsymbol{\theta}^{(1)} \\ \hline \boldsymbol{\theta}^{(2)} \end{array} \end{pmatrix}$, with $\boldsymbol{\theta}^{(1)} = \begin{pmatrix}
A \\ B \\ \alpha
\end{pmatrix}$ and $\boldsymbol{\theta}^{(2)} = \begin{pmatrix}
C \\ D \\ \beta
\end{pmatrix}. $ The parameter space can be written as  $\boldsymbol{\Theta}^{(1)} \times  \boldsymbol{\Theta}^{(2)}$ so that $\boldsymbol{\theta}^{(1)} \in  \boldsymbol{\Theta}^{(1)}$ and $\boldsymbol{\theta}^{(2)} \in \boldsymbol{\Theta}^{(2)},$ with $\boldsymbol{\Theta}^{(1)} = \boldsymbol{\Theta}^{(2)} = [-M, M] \times [-M, M] \times [0, \pi]$.  \\
\justify
Following is the algorithm to find the sequential estimators: \\
\begin{description}
\item \namedlabel{itm:step1} {\textbf{Step 1:}}\  First minimise the following error sum of squares:
\begin{equation}\label{ess_1}
Q_1(\boldsymbol{\theta}^{(1)}) = (\boldsymbol{Y} - \boldsymbol{Z}^{(1)}(\alpha)\boldsymbol{\mu}^{(1)})^{T}(\boldsymbol{Y} - \boldsymbol{Z}^{(1)}(\alpha)\boldsymbol{\mu}^{(1)})
\end{equation}
with respect to $A$, $B$ and $\alpha$. 
Using separable linear regression technique, for fixed $\alpha$, we have:
\begin{equation}\label{linear_est_sin_1}
\tilde{\boldsymbol{\mu}}^{(1)}(\alpha) = [\boldsymbol{Z}^{(1)}(\alpha)^{T}\boldsymbol{Z}^{(1)}(\alpha)]^{-1}\boldsymbol{Z}^{(1)}(\alpha)^{T}\boldsymbol{Y}.
\end{equation} 
Now replacing $\boldsymbol{\mu}^{(1)}$ by $\tilde{\boldsymbol{\mu}}^{(1)}(\alpha)$ in  \eqref{ess_1}, we have:
\begin{equation*}
R_1(\alpha) = Q_1(\tilde{A}, \tilde{B}, \alpha) = \boldsymbol{Y}^{T}(\boldsymbol{I} - \boldsymbol{Z}^{(1)}(\alpha)[\boldsymbol{Z}^{(1)}(\alpha)^{T}\boldsymbol{Z}^{(1)}(\alpha)]^{-1}\boldsymbol{Z}^{(1)}(\alpha)^T)\boldsymbol{Y}.
\end{equation*}
Minimising $R_1(\alpha)$, we obtain $\tilde{\alpha}$ and replacing $\alpha$ by $\tilde{\alpha}$ in  \eqref{linear_est_sin_1}, we get the linear parameter estimates $\tilde{A}$ and $\tilde{B}$.\\
\item {\textbf{Step 2:}}\namedlabel{itm:step2}\ At this step, we eliminate the effect of the sinusoid  component from the original data, and obtain a new data vector:
$$\boldsymbol{Y}_1 = \boldsymbol{Y} - \boldsymbol{Z}^{(1)}(\tilde{\alpha})\tilde{\boldsymbol{\mu}}^{(1)}.$$
Now we minimise the error sum of squares:
\begin{equation}\label{ess_2}
Q_2(\boldsymbol{\theta}^{(2)})= (\boldsymbol{Y}_1 - \boldsymbol{Z}^{(2)}(\beta)\boldsymbol{\mu}^{(2)})^{T}(\boldsymbol{Y}_1 - \boldsymbol{Z}^{(2)}(\beta)\boldsymbol{\mu}^{(2)}),
\end{equation}
with respect to $C$, $D$ and $\beta$. Again by separable linear regression technique, we have:
\begin{equation}\label{linear_est_chirp_1}
\tilde{\boldsymbol{\mu}}^{(2)}(\beta) = [\boldsymbol{Z}^{(2)}(\beta)^{T}\boldsymbol{Z}^{(2)}(\beta)]^{-1}\boldsymbol{Z}^{(2)}(\beta)^{T}\boldsymbol{Y}_{1}
\end{equation}
for a fixed $\beta$. Now replacing $\boldsymbol{\mu}^{(2)}$ by $\tilde{\boldsymbol{\mu}}^{(2)}$ in  \eqref{ess_2}, we obtain:
\begin{equation*}
R_2(\beta) = Q_2(\tilde{C}, \tilde{D}, \beta) = \boldsymbol{Y}_1^{T}(\boldsymbol{I} - \boldsymbol{Z}^{(2)}(\beta)[\boldsymbol{Z}^{(2)}(\beta)^{T}\boldsymbol{Z}^{(2)}(\beta)]^{-1}\boldsymbol{Z}^{(2)}(\beta)^T)\boldsymbol{Y}_1.
\end{equation*}
Minimizing $R_2(\beta)$, with respect to $\beta$, we obtain $\tilde{\beta}$, and using $\tilde{\beta}$ in \eqref{linear_est_chirp_1}, we obtain $\tilde{C}$  and $\tilde{D}$, the linear parameter estimates.
\end{description}
Note that instead of solving a 2D optimisation problem, as required to obtain the LSEs, to find the sequential estimators we need to solve two 1D optimisation problems. Also, the following theorems show that these estimators are strongly consistent and have the same asymptotic distribution as the LSEs.  
\begin{theorem}\label{consistency_sequential_LSEs_one_comp}
Under assumptions \ref{assump:1} and \ref{assump:2}, $\tilde{\boldsymbol{\theta}}^{(1)} = \begin{pmatrix}
\tilde{A} & \tilde{B} & \tilde{\alpha} \end{pmatrix}$ and  $\tilde{\boldsymbol{\theta}}^{(2)} = \begin{pmatrix}
\tilde{C} & \tilde{D} & \tilde{\beta} \end{pmatrix}$ are strongly consistent estimators of ${\boldsymbol{\theta}^0}^{(1)} = \begin{pmatrix}
A^0 & B^0 & \alpha^0 \end{pmatrix}$ and ${\boldsymbol{\theta}^0}^{(2)} = \begin{pmatrix}
C^0 & D^0 & \beta^0 \end{pmatrix}$, respectively, that is,
\begin{enumerate}[label=(\alph*)]
\item $\tilde{\boldsymbol{\theta}}^{(1)} \xrightarrow{a.s.} {\boldsymbol{\theta}^0}^{(1)} $ as $n \rightarrow \infty$,
\item $\tilde{\boldsymbol{\theta}}^{(2)} \xrightarrow{a.s.} {\boldsymbol{\theta}^0}^{(2)}$ as $n \rightarrow \infty$.
\end{enumerate}
\end{theorem} 
\begin{proof}
See Appendix ~\ref{appendix:C2}.\\
\end{proof}
\begin{theorem}\label{asymptotic_dist_sequential_LSEs_one_comp}
Under assumptions \ref{assump:1} and \ref{assump:2},
\begin{enumerate}[label=(\alph*)]
\item 
$(\tilde{\boldsymbol{\theta}}^{(1)} - {\boldsymbol{\theta}^0}^{(1)})\textbf{D}_1^{-1} \xrightarrow{d} \mathcal{N}_3(0, \sigma^2c{\boldsymbol{\Sigma}^{(1)}}^{-1})$,
\item $(\tilde{\boldsymbol{\theta}}^{(2)} - {\boldsymbol{\theta}^0_1}^{(2)})\textbf{D}_2^{-1} \xrightarrow{d} \mathcal{N}_3(0, \sigma^2c{\boldsymbol{\Sigma}^{(2)}}^{-1})$,
\end{enumerate}
where $\textbf{D}_1$ and $\textbf{D}_2$, are sub-matrices of the diagonal matrix $\textbf{D}$ such that $\textbf{D} = \begin{pmatrix}\begin{array}{c|c}
\textbf{D}_1 & \mathbf{0}\\ \hline
\mathbf{0} & \textbf{D}_2\\
\end{array}
\end{pmatrix}.$ Similarly, ${\boldsymbol{\Sigma}^{(1)}}^{-1}$ and ${\boldsymbol{\Sigma}^{(2)}}^{-1}$ are such that $\boldsymbol{\Sigma}^{-1}(\boldsymbol{\theta}^0) = \begin{pmatrix}\begin{array}{c|c}
{\boldsymbol{\Sigma}^{(1)}}^{-1}(\boldsymbol{\theta}^0) & \mathbf{0}\\ \hline
\mathbf{0} & {\boldsymbol{\Sigma}^{(2)}}^{-1}(\boldsymbol{\theta}^0)\\
\end{array}
\end{pmatrix}.$ Note that, $\textbf{D}$, $c$ and $\boldsymbol{\Sigma}^{-1}(\boldsymbol{\theta}^0)$ are as defined in Theorem \ref{asymptotic_dist_LSEs_one_component}.
\end{theorem}
\begin{proof}
See Appendix ~\ref{appendix:C2}.\\
\end{proof}

\section{Multiple Component Chirp-like Model}\label{multiple_component_theory}
To model real life phenomena effectively, we require a more adaptable model. In this section, we consider a multiple component chirp-like model, defined in \eqref{multiple_comp_model}, a natural generalisation of the one component model. 
Under certain assumptions  in addition to Assumption \ref{assump:1} on the noise component, that we state below, we study the asymptotic properties of the LSEs and provide the results in the following subsection.
\justify
Let us denote $\boldsymbol{\vartheta}$ as the parameter vector for model \eqref{multiple_comp_model},
$$\boldsymbol{\vartheta} =
(A_1, B_1, \alpha_1,  \cdots, A_p, B_p, \alpha_p, C_1, D_1, \beta_1, \cdots, C_q, D_q, \beta_q).$$
Also, let $\boldsymbol{\vartheta}^0$ denote the true parameter vector and $\hat{\boldsymbol{\vartheta}}$, the LSE of $\boldsymbol{\vartheta}^0.$                                                                                                                                                                                                                                                                                                                                                                                                                                                                                                                                                                                                                                                                                                                                                                                                                                                                                                                                                                                                                                                                                                                                                                                                                                                                   

\begin{assumption}\label{assump:3}
$\boldsymbol{\vartheta}^0$ is an interior point of $\bm{\mathcal{V}} = {\boldsymbol{\Theta}_{1}}^{(p+q)} $, the parameter space and the frequencies $\alpha_{j}^0s$ are distinct for $j = 1, \cdots p$ and so are the frequency rates $\beta_{k}^0s$  for $k = 1, \cdots q$. Note that $\boldsymbol{\Theta}_{1} = [-M, M] \times [-M, M]  \times [0,\pi]. $  \\                                                                                                                                                                                                                                                                                                                                                                                                                                                                                                                                                                                                                                                                                                                                                                                                                                                                                                                                                                                                                                                                                                                                                                                                                                                                            \end{assumption}
\begin{assumption}\label{assump:4}
The amplitudes, $A_j^0$s and $B_j^0$s satisfy the following relationship: \\
$$ \infty > {A_{1}^{0}}^2 + {B_{1}^{0}}^2 > {A_{2}^{0}}^2 + {B_{2}^{0}}^2 > \cdots > {A_{p}^{0}}^2 + {B_{p}^{0}}^2 > 0. $$
Similarly, $C_k^0$s and $D_k^0$s satisfy the following relationship: \\
$$ \infty > {C_{1}^{0}}^2 + {D_{1}^{0}}^2 > {C_{2}^{0}}^2 + {D_{2}^{0}}^2 > \cdots > {C_{q}^{0}}^2 + {D_{q}^{0}}^2 > 0. $$
\end{assumption}
\subsection{Least Squares Estimators}\label{multiple_component_LSEs_theory}
The LSEs of the unknown parameters of the proposed model (see  \eqref{multiple_comp_model}), can be obtained by minimising the error sum of squares:
\begin{equation}\label{ess_multiple_component}
Q(\boldsymbol{\vartheta}) = \sum_{t=1}^{n}\bigg(y(t) - \sum_{j=1}^{p} A_j \cos(\alpha_j t) + B_j \sin(\alpha_j t) - \sum_{k=1}^{q} C_k \cos(\beta_k t^2) + D_k \sin(\beta_k t^2)\bigg)^2,
\end{equation}
with respect to $A_1$, $B_1$, $\alpha_1$, $\cdots$,  $A_p$, $B_p$ $\alpha_p$, $C_1$,$D_1$,$\beta_1$, $\cdots$, $C_q$  $D_q$ and $\beta_q$ simultaneously.  Similar to the one component model, $Q(\boldsymbol{\vartheta})$ can be expressed in matrix notation and then the LSE,  $\hat{\boldsymbol{\vartheta}}$ of $\boldsymbol{\vartheta}^0$, can be obtained along the similar lines.
\justify
Next we examine the consistency property of the LSE $\hat{\boldsymbol{\vartheta}}$ along with its asymptotic distribution.
 
\begin{theorem}\label{consistency_LSEs_multiple_comp}
If assumptions, \ref{assump:1}, \ref{assump:3} and \ref{assump:4}, hold true, then:
$$\hat{\boldsymbol{\vartheta}} \xrightarrow{a.s.} \boldsymbol{\vartheta}^0 \textmd{ as } n \rightarrow \infty.$$
\end{theorem}
\begin{proof}
The consistency of the LSE $\hat{\boldsymbol{\vartheta}}$ can be proved along the similar lines as the consistency of the LSE $\hat{\boldsymbol{\theta}}$, for the one component model.\\
\end{proof}

\begin{theorem}\label{asymptotic_distribution_LSEs_multiple_comp}
If assumptions, \ref{assump:1}, \ref{assump:3} and \ref{assump:4}, hold true, then:
$$(\hat{\boldsymbol{\vartheta}} - \boldsymbol{\vartheta}^0)\mathfrak{D}^{-1} \xrightarrow{d} \mathcal{N}_{3(p+q)}(0, \mathcal{E}^{-1}(\boldsymbol{\vartheta}^0)).$$
Here, $\mathfrak{D} = diag(\underbrace{\textbf{D}_1, \cdots \textbf{D}_1}_{p\ times}, \underbrace{\textbf{D}_2, \cdots, \textbf{D}_2}_{q\ times})$, where $\textbf{D}_1 = diag(\frac{1}{\sqrt{n}}, \frac{1}{\sqrt{n}}, \frac{1}{n\sqrt{n}})$ and $\textbf{D}_2 = diag(\frac{1}{\sqrt{n}}, \frac{1}{\sqrt{n}}, \frac{1}{n^2\sqrt{n}})$.\\
$$  \mathcal{E}(\boldsymbol{\vartheta}^0)  = \begin{pmatrix}
\boldsymbol{\Sigma}^{(1)}_1 & 0              &   \cdots   &  & \cdots & 0 &   \\
0              & \ddots &   0  &    & \cdots &  0 &         \\
\vdots         & \vdots         &  \boldsymbol{\Sigma}^{(1)}_{p}  & 0 & \cdots & 0  \\
0              &   \cdots           &  0      &  \boldsymbol{\Sigma}^{(2)}_1 & 0 & 0  \\
0              &    \cdots           & \cdots     &    0            & \ddots & 0  \\
0              &    \cdots          &      &    \cdots      &   0 & \boldsymbol{\Sigma}^{(2)}_{q}   \\
\end{pmatrix},$$
with $\boldsymbol{\Sigma}^{(1)}_j = \begin{pmatrix}
\frac{1}{2} & 0 & \frac{B_j^0}{4} \\
0 & \frac{1}{2} & \frac{-A_j^0}{4} \\
\frac{B_j^0}{4} & \frac{-A_j^0}{4} & \frac{{A_j^0}^2 + {B_j^0}^2}{6} \\
\end{pmatrix}, j = 1, \cdots, p $ and $\boldsymbol{\Sigma}^{(2)}_k = \begin{pmatrix}
\frac{1}{2} & 0 & \frac{D_k^0}{6} \\
0 & \frac{1}{2} & \frac{-C_k^0}{6} \\
\frac{D_k^0}{6} & \frac{-C_k^0}{6} & \frac{{C_k^0}^2 + {D_k^0}^2}{10} \\
\end{pmatrix}, k = 1, \cdots, q.$  
\end{theorem}
\begin{proof}
See Appendix ~\ref{appendix:D1}.\\
\end{proof}

\subsection{Sequential Least Squares Estimators}\label{multiple_component_seq_LSEs_theory}
For the multiple component chirp-like model, if the number of components, $p$ and $q$ are very large, finding the LSEs becomes computationally challenging. To resolve this issue, we propose a sequential procedure to estimate the unknown parameters  similar to  the one component model. Using the sequential procedure, the $(p+q)-$dimensional optimisation problem can be reduced to $p+q$, 1D optimisation problems. The algorithm for the sequential estimation is as follows:\\
\justify
\begin{itemize}
\item Perform \textbf{Step 1} (see section \ref{one_component_seq_LSEs_theory}) and obtain the estimate, $\tilde{\boldsymbol{\theta}}^{(1)}_1$ =  $(\tilde{A}_1$, $\tilde{B}_1$, $\tilde{\alpha}_1$).
\item Eliminate the effect of the estimated sinusoid component and obtain new data vector:
$$y_1(t) = y(t) - \tilde{A}_1 \cos(\tilde{\alpha}_1 t) -  \tilde{B}_1 \sin(\tilde{\alpha}_1 t). $$
\item Minimize the following error sum of squares to obtain the estimates of the next sinusoid, $\tilde{\boldsymbol{\theta}}^{(1)}_2$ =  $(\tilde{A}_2$, $\tilde{B}_2$, $\tilde{\alpha}_2$):
$$Q_2(A,B,\alpha) = \sum_{t=1}^{n}\bigg(y_1(t) - A \cos(\alpha t) - B \sin(\alpha t)\bigg)^2.$$
\end{itemize}
Repeat these steps until all the $p$ sinusoids are estimated.
\begin{itemize}[resume]
\item Finally at $(p+1)$-$th$ step, we obtain the data:
$$y_p(t) = y_{p-1}(t) - \tilde{A}_p \cos(\tilde{\alpha}_p t) - \tilde{B}_p \sin(\tilde{\alpha}_p t).$$
\item Using this data we estimate the first chirp-like component parameters, and obtain  $\tilde{\boldsymbol{\theta}}^{(2)}_1$ =  $(\tilde{C}_1$, $\tilde{D}_1$, $\tilde{\beta}_1$): by minimizing:
$$Q_{p+1}(C,D,\beta) = \sum_{t=1}^{n}\bigg(y_{p}(t) - C\cos(\beta t^2) - D \sin(\beta t^2)\bigg)^2.$$ 
\item Now eliminate the effect of this estimated chirp component and obtain: $y_{p+1}(t) = y_{p}(t) - \tilde{C}_1 \cos(\tilde{\beta}_1 t^2) - \tilde{D}_1 \sin(\tilde{\beta}_1 t^2)$ and minimize $Q_{p+2}(C, D, \beta)$ to obtain  $\tilde{\boldsymbol{\theta}}^{(2)}_2$ =  $(\tilde{C}_2$, $\tilde{D}_2$, $\tilde{\beta}_2$).
\end{itemize} 
Continue to do so and estimate all the $q$ chirp components.

\justify
Further in this section, we examine the consistency property of the proposed sequential  estimators, when $p$ and $q$ are unknown. Thus, we consider the following two cases: (a) when the number of components of the fitted model is  less than the actual number of components, and (b) when the number of components of the fitted model is more than the actual number of components.

\begin{theorem}\label{consistency_first component}
If assumptions \ref{assump:1}, \ref{assump:3} and \ref{assump:4} are satisfied, then the following are true:
\begin{enumerate}[label=(\alph*)]
\item $\tilde{\boldsymbol{\theta}}_1^{(1)} \xrightarrow{a.s.} {\boldsymbol{\theta}_1^0}^{(1)}$ as $n \rightarrow \infty,$
\item $\tilde{\boldsymbol{\theta}}_1^{(2)} \xrightarrow{a.s.} {\boldsymbol{\theta}_1^0}^{(2)}$ as $n \rightarrow \infty.$
\end{enumerate}
\end{theorem}
\begin{proof}
See Appendix ~\ref{appendix:D2}.\\
\end{proof}

\begin{theorem}\label{consistency_rest_of_the_components}
If assumptions \ref{assump:1}, \ref{assump:3} and \ref{assump:4} are satisfied, the following are true:
\begin{enumerate}[label=(\alph*)]
\item  $\tilde{\boldsymbol{\theta}}_j^{(1)} \xrightarrow{a.s.} {\boldsymbol{\theta}_j^0}^{(1)}$ as $n \rightarrow \infty$, for all $j = 2, \cdots, p$,

\item $\tilde{\boldsymbol{\theta}}_k^{(2)} \xrightarrow{a.s.} {\boldsymbol{\theta}_k^0}^{(2)}$ as $n \rightarrow \infty$, for all $k = 2, \cdots, q$.
\end{enumerate}
\end{theorem}
\begin{proof}
See Appendix ~\ref{appendix:D2}.\\
\end{proof}

\begin{theorem}\label{consistency_excess_components}
If assumptions \ref{assump:1}, \ref{assump:3} and \ref{assump:4} are true, then the following are true:
\begin{enumerate}[label=(\alph*)]
\item $\tilde{A}_{p+1} \xrightarrow{a.s.}$ 0, $\tilde{B}_{p+1} \xrightarrow{a.s.} 0$,
\item $\tilde{C}_{q+1} \xrightarrow{a.s.}$ 0, $\tilde{D}_{q+1} \xrightarrow{a.s.} 0$.
\end{enumerate}
\end{theorem}
\begin{proof}
See Appendix ~\ref{appendix:D2}.\\
\end{proof}
\justify
Next we determine the asymptotic distribution of the proposed estimators at each step through the following theorems:

\begin{theorem}\label{asymptotic_distribution_first_component}
If assumptions \ref{assump:1}, \ref{assump:3} and \ref{assump:4} are satisfied, then:
\begin{enumerate}[label=(\alph*)]
\item $(\tilde{\boldsymbol{\theta}}_1^{(1)} - {\boldsymbol{\theta}^0_1}^{(1)})\textbf{D}_1^{-1} \xrightarrow{d} \mathcal{N}_3(0, \sigma^2c{\boldsymbol{\Sigma}^{(1)}_1}^{-1}),$
\item $(\tilde{\boldsymbol{\theta}}_1^{(2)} - {\boldsymbol{\theta}^0_1}^{(2)})\textbf{D}_2^{-1} \xrightarrow{d} \mathcal{N}_3(0, \sigma^2c{\boldsymbol{\Sigma}^{(2)}_1}^{-1}).$
\end{enumerate}
Here, $\boldsymbol{\Sigma}^{(1)}_1 = \begin{pmatrix}
\frac{1}{2} & 0 & \frac{B_1^0}{4} \\
0 & \frac{1}{2} & -\frac{A_1^0}{4} \\
\frac{B_1^0}{4}  & -\frac{A_1^0}{4} & \frac{{A_1^0}^2 + {B_1^0}^2}{6}
\end{pmatrix} $ , 
${\boldsymbol{\Sigma}^{(2)}_1} = \begin{pmatrix}
\frac{1}{2} & 0 & \frac{D_1^0}{6} \\
0 & \frac{1}{2} & -\frac{C_1^0}{6} \\
\frac{D_1^0}{6}  & -\frac{C_1^0}{6} & \frac{{C_1^0}^2 + {D_1^0}^2}{10}
\end{pmatrix}$  and $c = \sum\limits_{j = -\infty}^{\infty}a^2(j).$
\end{theorem} 
\begin{proof}
See Appendix ~\ref{appendix:D2}.\\
\end{proof}
\justify
This result can be extended for $j \leqslant p$ and $k \leqslant q$  as follows: 
\begin{theorem}\label{asymptotic_distribution_rest_of_the_components}
If assumptions \ref{assump:1}, \ref{assump:3} and \ref{assump:4} are satisfied, then:
\begin{enumerate}[label=(\alph*)]
\item $(\tilde{\boldsymbol{\theta}}_j^{(1)} - {\boldsymbol{\theta}^0_j}^{(1)})\textbf{D}_1^{-1} \xrightarrow{d} \mathcal{N}_3(0, \sigma^2c {\boldsymbol{\Sigma}^{(1)}_j}^{-1}),$
\item $(\tilde{\boldsymbol{\theta}}_k^{(2)} - {\boldsymbol{\theta}^0_k}^{(2)})\textbf{D}_2^{-1} \xrightarrow{d} \mathcal{N}_3(0, \sigma^2c {\boldsymbol{\Sigma}^{(2)}_k}^{-1}).$
\end{enumerate} 
$\boldsymbol{\Sigma}^{(1)}_j$ and $\boldsymbol{\Sigma}^{(2)}_k$ can be obtained from $\boldsymbol{\Sigma}^{(1)}_1$ and $\boldsymbol{\Sigma}^{(2)}_1$ respectively, by replacing $A_1$ and $B_1$ by  $A_j$ and $B_j$  and $C_1$ and $D_1$ by $C_k$ and $D_k$, respectively.
\end{theorem}
\begin{proof}
This can be obtained along the same lines as the proof of Theorem~\ref{asymptotic_distribution_first_component}.  \\
\end{proof}
\justify
From the above results, it is evident that the sequential LSEs are strongly consistent and have the same asymptotic distribution as the LSEs and at the same time can be computed more efficiently. Thus for the simulation studies and to analyse the real data sets as well, we compute the sequential LSEs instead of the LSEs.
\section{Simulation Studies}\label{simulation_results}
In this section, we first present numerical experiments for the one component chirp-like model:
$$y(t) = A^0 \cos(\alpha^0 t)+ B^0 \sin(\alpha^0 t)  + C^0 (\beta^0 t^2) + D^0 (\beta^0 t^2) + X(t),$$
with the true parameter values: $A^0 = 10$, $B^0 = 10$, $\alpha^0 = 1.5$, $C^0 = 10$, $D^0 = 10$ and $\beta^0 = 0.1$. We consider the following errors to generate data for these simulation experiments:
\begin{align}
& 1. \quad X(t) = \epsilon(t).\label{eq:error_normal}\\
& 2. \quad X(t) = \epsilon(t) + 0.5  \epsilon(t-1).\label{eq:error_MA}
\end{align}
\justify
Here $\epsilon(t)$s are i.i.d. normal random variables with mean zero and variance $\sigma^2$. We consider different sample sizes: $n = 100, 200, 300, 400$ and $500$  and different error variances: $\sigma^2$: 0.1, 0.5 and 1. For each $n$ we replicate the process, that is generate the data and obtain the proposed sequential LSEs 1000 times and report their averages, biases and variances. We also compute the theoretical asymptotic variances of the proposed estimators to compare with the corresponding variances. The results are provided in Table \ref{table2}- Table \ref{table6} in Appendix~\ref{appendix:A1}. From these results, it can be seen that the average estimates are quite close to the true values and the biases are very small. As $n$ increases, the biases and the variances, decrease and as the error variance increases, they decrease. Also, the variances and the asymptotic variances match well, particularly when $n$ increases.\\ \par

In Appendix \ref{appendix:A2}, we present the simulation studies for the multiple component chirp-like model with $p = q = 2$. Following are the true parameter values used for data generation:
$$A_1^0 = 10, B_1^0 = 10, \alpha_1^0 = 1.5, C_1^0 = 10, D_1^0 = 10 \textmd{ and } \beta_1^0 = 0.1,$$
$$A_2^0 = 8, B_2^0 = 8, \alpha_2^0 = 2.5, C_2^0 = 8, D_2^0 = 8 \textmd{ and } \beta_2^0 = 0.2.$$
The error structure used for data generation are same as used for the one component model simulations, \eqref{eq:error_normal} and \eqref{eq:error_MA}. We compute the sequential LSEs and report their averages, biases, variances and asymptotic variances. Again, the different sample sizes and error variances considered, are same as that for the one component model. Tables \ref{table7}-\ref{table11}, provide the results obtained. It is observed that the estimates obtained have significantly small biases and are close to the true values. Moreover as $n$ increases the variances decrease, thus depicting the desired consistency of the estimators. Also the order of the variances is more or less same as that of the corresponding asymptotic variances. Thus, they are well matched. This validates the explicitly derived theoretical asymptotic variances of the proposed estimators. 
\section{Data Analysis}\label{data_analysis}
In this section, we analyse  four different speech signal data sets: "AAA", "AHH", "UUU" and "EEE" using chirp model as well as the proposed chirp-like model. These data sets have been obtained from a sound instrument at the Speech Signal Processing laboratory of the Indian Institute of Technology Kanpur. The data set "AAA", has 477 data points, the set "AHH" has 469 data points and the rest of them have 512 points each. \\ \par

We fit the chirp-like model to these data sets using the sequential LSEs following the algorithm  described in section \ref{multiple_component_seq_LSEs_theory}. As is evident from the description, we need to solve a 1D optimisation problem to find these estimators and since the problem is nonlinear, we need to employ some iterative method to do so. Here we use Brent's method to solve 1D optimisation problems, using an inbuilt function in R, known as optim. For this method to work, we require very good initial values in the sense that they need to be close to the true values. Now one of the well received methods for finding initial values   for the frequencies of the sinusoidal model is to maximize the periodogram function: 
$$I_1(\alpha) = \frac{1}{n} \bigg|\sum_{t=1}^{n} y(t)e^{-i\alpha t}\bigg|^2 $$
at the points: $\displaystyle{\frac{\pi j}{n}}$; $j = 1, \cdots, n-1,$ called the Fourier frequencies. The estimators obtained by this method, are called the Periodogram Estimators. After all the $p$ sinusoid components are fitted, we need to fit the $q$ chirp components. Again we need to solve 1D optimisation problem at each stage and for that we need good initial values. Analogous to the periodogram function $I_1(\alpha)$, we define a periodogram-type function as follows:
$$I_2(\beta) = \frac{1}{n} \bigg|\sum_{t=1}^{n} y(t)e^{-i\beta t^2}\bigg|^2.$$
To obtain the starting points for the frequency rate parameter $\beta$, we maximise this function at the points: $\displaystyle{\frac{\pi k}{n^2}}$; $j = 1, \cdots, n^2-1$, similar to the Fourier frequencies.\\ \par
Since in practice, the number of components of a model are unknown, we need to estimate them. We use the following Bayesian Information Criterion (BIC) criterion, as a tool to estimate $p$ and $q$:
$$\textmd{BIC}(p,q) = n \ ln(\textmd{SSE}(p,q)) + 2\ (3p + 3q)\ ln(n)$$
for the present analysis of the data sets.\\ \par
For comparison of the chirp-like model with the chirp model, we re-analyse these data sets by fitting chirp model to each of them (for methodology, see Lahiri et al. \cite{2015} and Grover et al. \cite{2018}). In the following table, we report the number of components required to fit the chirp model and the chirp-like model to each of the data sets and in the subsequent figures, we plot the original data along with the estimated signals obtained by fitting a chirp model and a chirp-like model to these data. In both scenarios, the model is fitted using the sequential LSEs.
\begin{table}[H]
\centering
\begin{tabularx}{0.75\textwidth}{|*{5}{C|}}
\hline
\multicolumn{5}{|>{\hsize=5\hsize}C|}{Number of components}                                                            \\ \hline
\multirow{2}{*}{Data Set}  & \multicolumn{2}{>{\hsize=2\hsize}C|}{Chirp Model} & \multicolumn{2}{>{\hsize=2\hsize}C|}{Chirp-like model} \\ \cline{2-5} 
                           & p & Number of parameters & (p, q) & Number of parameters               \\ \hline
AAA                        & 9 & 36  & (10, 1)& 33           \\ \hline
AHH                        & 8 & 32  & (7, 1) & 24           \\ \hline
UUU                        & 9 & 36  & (8, 1) & 27           \\ \hline
EEE                        &11 & 44  & (14, 1)& 45           \\ \hline
\end{tabularx}
\caption{Number of components used to fit chirp and chirp-like model to the speech data sets.}
\label{number_of_comp}
\end{table}
\begin{figure}[H]
\centering
\includegraphics[scale = 0.35]{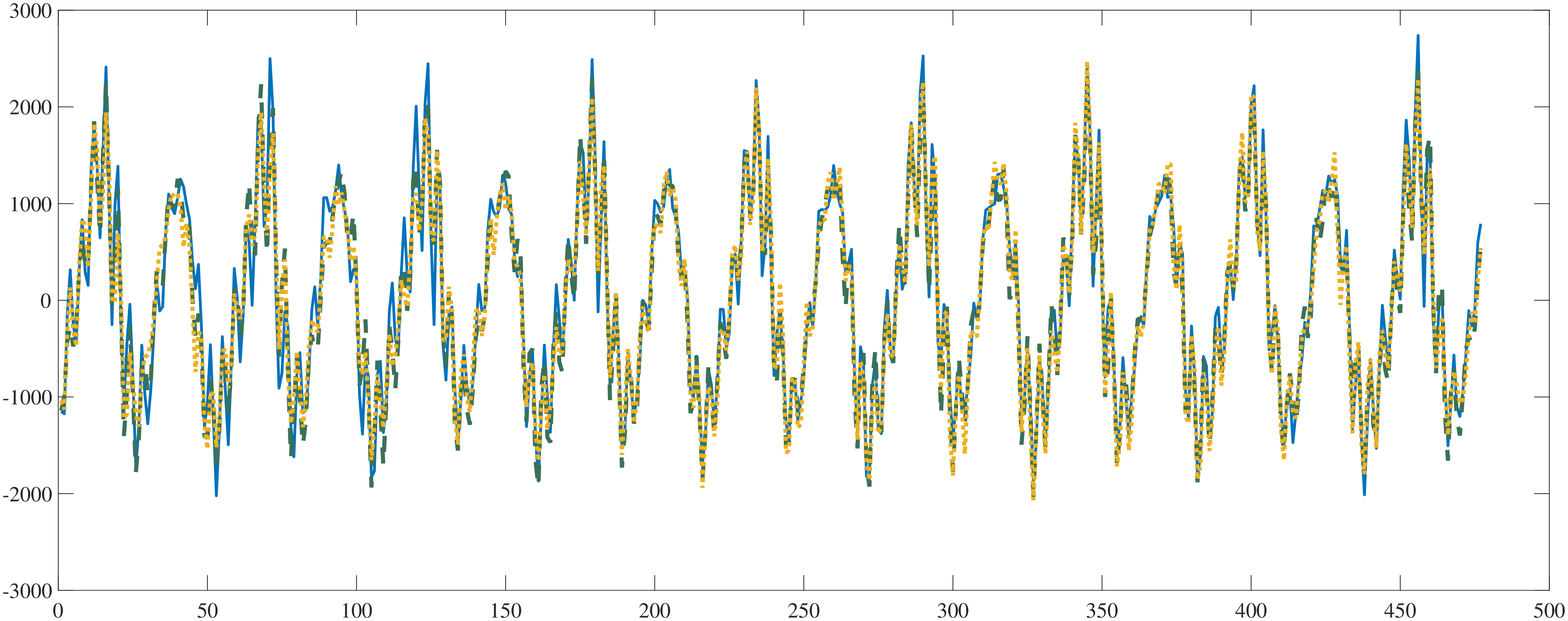}
\caption{Observed data "AAA" (blue solid line) and fitted chirp signal (green dashed line) and fitted chirp-like signal (yellow dotted line) to the data}
\label{fig_fittings_AAA}
\end{figure}
\begin{figure}[H]
\centering
\includegraphics[scale = 0.35]{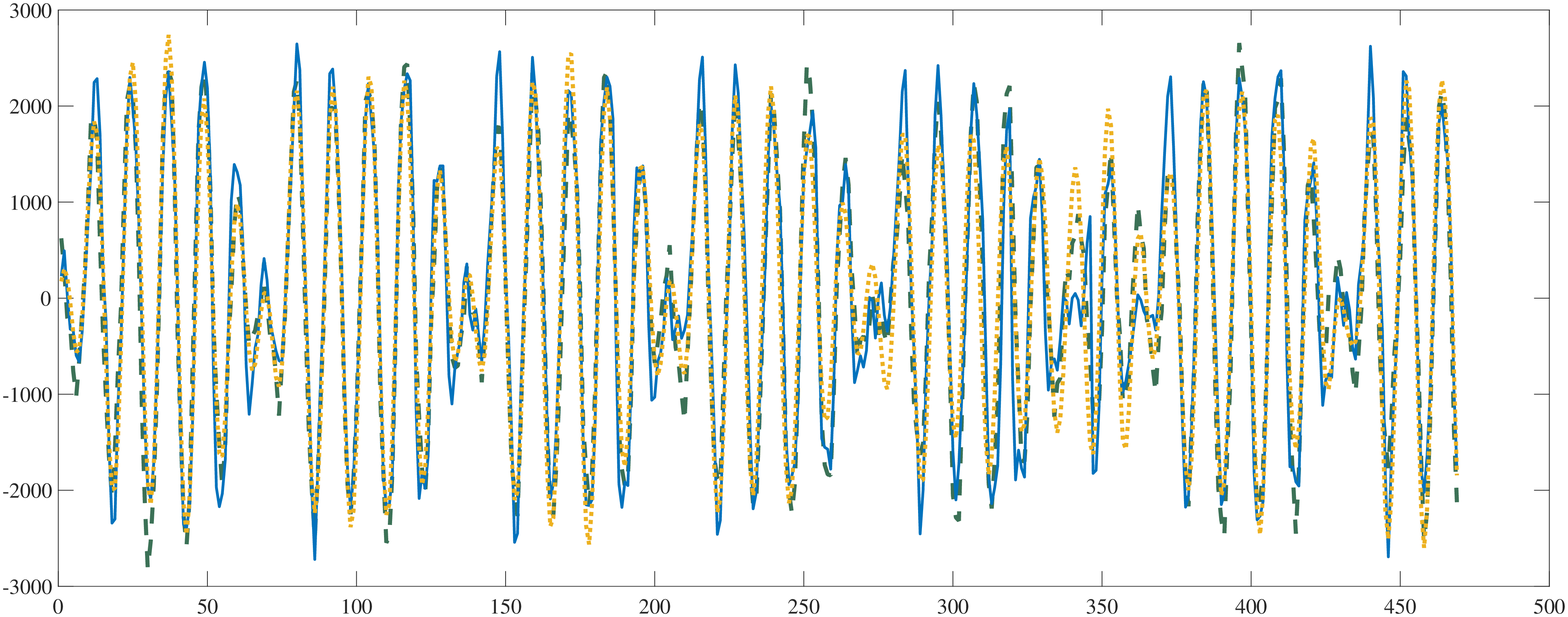}
\caption{Observed data "AHH" (blue solid line) and fitted chirp signal (green dashed line) and fitted chirp-like signal (yellow dotted line) to the data}
\label{fig_fittings_AHH}
\end{figure}  
\begin{figure}[H]
\centering
\includegraphics[scale = 0.35]{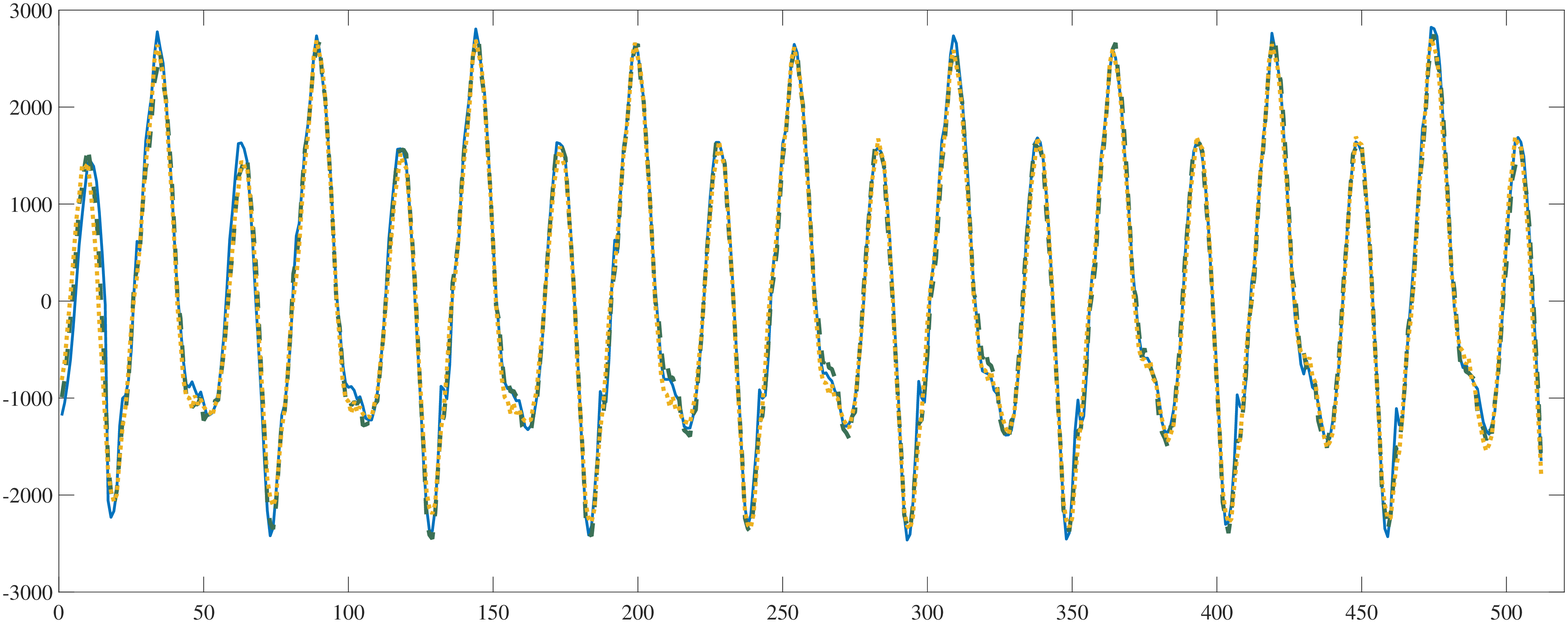}
\caption{Observed data "UUU" (blue solid line) and fitted chirp signal (green dashed line) and fitted chirp-like signal (yellow dotted line) to the data}
\label{fig_fittings_UUU}
\end{figure}
\begin{figure}[H]
\centering
\includegraphics[scale = 0.35]{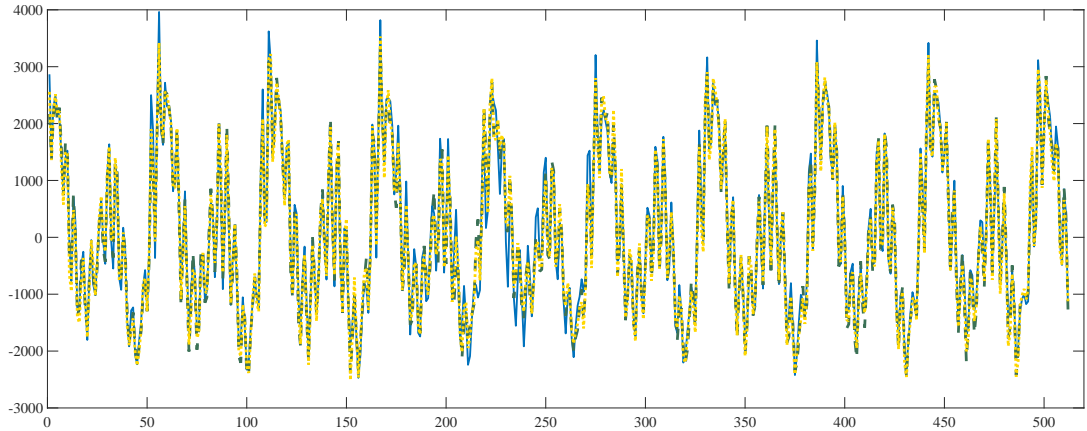}
\caption{Observed data "EEE" (blue solid line) and fitted chirp signal (green dashed line) and fitted chirp-like signal (yellow dotted line) to the data}
\label{fig_fittings_EEE}
\end{figure}  
\justify
To validate the error assumption of stationarity, we test the residuals, for all the cases, using augmented Dickey Fuller test. We use an inbuilt function adftest in MATLAB for this purpose. The test statistic values result in rejection of the null hypothesis of a presence of a unit root indicating that residuals, in all the cases, are stationary.\\ \par
\justify
It is evident from the figures above, that visually both the models provide a good fit for all the speech data sets. However, to fit a chirp-like model using the sequential LSEs, we solve a 1D optimisation problem at each step while for the fitting of a chirp model, at each step we need to deal with a 2D optimisation problem. Moreover, to find the initial values, in both cases, a grid search is performed and for the chirp-like model, this means evaluation of the periodogram functions $I_1(\alpha)$ and $I_2(\beta)$ at $n^2$ and $n$ grid points respectively as opposed to the $n^3$ grid points for the chirp model. Note that this is done at each step for the sequential estimators and hence becomes more complex as the number of components increase. Thus fitting a chirp-like model is numerically more efficient than fitting a chirp model.

\section{Conclusion}\label{conclusion}
Chirp signals are ubiquitous in many areas of science and engineering and hence their parameter estimation is of great significance in signal processing. But it has been observed that parameter estimation of this model, particularly using the method of least squares is computationally complex. In this paper, we put forward an alternate model, named \textit{chirp-like} model. We observe that the data that have been analysed using chirp models can also be analysed using the chirp-like model and estimating its parameters using sequential LSEs is simpler than that for the chirp model. We show that the LSEs and the sequential LSEs of the parameters of this model are strongly consistent and asymptotically normally distributed. The rates of convergence of the parameter estimates of this model are same as those for the chirp model. We analyse four speech data sets, and it is observed that the proposed model can be used quite effectively to analyse these data sets.

\begin{appendices}
\section{Numerical Results}\label{appendix:A}
\subsection{One component chirp-like model}\label{appendix:A1}
\begin{table}[H]
\begin{minipage}{.5\linewidth}
\caption{LSEs when sample size is 100}
\centering
\resizebox{0.95\textwidth}{!}{\begin{tabular}{|c|c|c|c|c|c|}
\hline
\multicolumn{2}{|c|}{Type of error}            & \multicolumn{2}{c|}{N(0,$\sigma^2$)}          & \multicolumn{2}{c|}{MA(1) with $\rho = 0.5$}  \\ \hline
\multicolumn{2}{|c|}{Parameters}               & $\alpha$              & $\beta$               & $\alpha$              & $\beta$               \\ \hline
\multicolumn{2}{|c|}{True values}              & 1.5                   & 0.1                   & 1.5                   & 0.1                   \\ \hline
$\sigma^2$             & \multicolumn{1}{l|}{} & \multicolumn{4}{c|}{Least Squares Estimates}                                                  \\ \hline
0.10         &  Average                        &  1.4960      &     0.1000     &      1.4960     &     0.1000   \\  \hline
             &  Bias                           &  -3.98e-03   &    -1.62e-06   &     -3.97e-03   &    -1.67e-06 \\  \hline 
             &  Variance                       &  1.02e-08    &     1.25e-12   &      1.54e-08   &     1.83e-12 \\  \hline
             &  Asym Var                       &  1.20e-08    &     1.12e-12   &      1.50e-08   &     1.41e-12 \\  \hline
\multicolumn{1}{|l|}{} & \multicolumn{1}{l|}{} & \multicolumn{1}{l|}{} & \multicolumn{1}{l|}{} & \multicolumn{1}{l|}{} & \multicolumn{1}{l|}{} \\ \hline
0.50         &  Average                        &  1.4960      &     0.1000     &      1.4960     &     0.1000   \\  \hline
             &  Bias                           &  -3.97e-03   &    -1.65e-06   &     -3.97e-03   &    -1.56e-06 \\  \hline
             &  Variance                       &  5.41e-08    &     5.94e-12   &      7.21e-08   &     9.45e-12 \\  \hline
             &  Asym Var                       &  6.00e-08    &     5.62e-12   &      7.50e-08   &     7.03e-12 \\  \hline
\multicolumn{1}{|l|}{} & \multicolumn{1}{l|}{} & \multicolumn{1}{l|}{} & \multicolumn{1}{l|}{} & \multicolumn{1}{l|}{} & \multicolumn{1}{l|}{} \\ \hline
1.00         &  Average                        &  1.4960      &     0.1000     &      1.4960     &     0.1000   \\  \hline
             &  Bias                           &  -3.99e-03   &    -1.83e-06   &     -3.97e-03   &    -1.61e-06 \\  \hline
             &  Variance                       &  1.02e-07    &     1.25e-11   &      1.36e-07   &     1.80e-11 \\  \hline
             &  Asym Var                       &  1.20e-07    &     1.12e-11   &      1.50e-07   &     1.41e-11 \\  \hline
\end{tabular}}
\label{table2}
\end{minipage}%
\begin{minipage}{.5\linewidth}
\caption{LSEs when sample size is 200}
\centering
\resizebox{0.95\textwidth}{!}{\begin{tabular}{|c|c|c|c|c|c|}
\hline
\multicolumn{2}{|c|}{Type of error}            & \multicolumn{2}{c|}{N(0,$\sigma^2$)}          & \multicolumn{2}{c|}{MA(1) with $\rho = 0.5$}  \\ \hline
\multicolumn{2}{|c|}{Parameters}               & $\alpha$              & $\beta$               & $\alpha$              & $\beta$               \\ \hline
\multicolumn{2}{|c|}{True values}              & 1.5                   & 0.1                   & 1.5                   & 0.1                   \\ \hline
$\sigma^2$             & \multicolumn{1}{l|}{} & \multicolumn{4}{c|}{Least Squares Estimates}                                                  \\ \hline
0.10         &  Average                        &  1.4984      &    0.1000     &    1.4984      &    0.1000    \\  \hline
             &  Bias                           &  -1.56e-03   &   -2.30e-07   &   -1.56e-03    &   -2.22e-07  \\  \hline 
             &  Variance                       &  1.34e-09    &    3.74e-14   &    1.85e-09    &    4.48e-14  \\  \hline
             &  Asym Var                       &  1.50e-09    &    3.52e-14   &    1.87e-09    &    4.39e-14  \\  \hline
\multicolumn{1}{|l|}{} & \multicolumn{1}{l|}{} & \multicolumn{1}{l|}{} & \multicolumn{1}{l|}{} & \multicolumn{1}{l|}{} & \multicolumn{1}{l|}{} \\ \hline
0.50         &  Average                        &   1.4984    &    0.1000     &    1.4984       &     0.1000   \\  \hline
             &  Bias                           &   -1.56e-03 &   -2.04e-07   &   -1.56e-03     &    -2.31e-07 \\  \hline
             &  Variance                       &   6.69e-09  &    1.83e-13   &    8.91e-09     &     2.38e-13 \\  \hline
             &  Asym Var                       &   7.50e-09  &    1.76e-13   &    9.37e-09     &     2.20e-13 \\  \hline
\multicolumn{1}{|l|}{} & \multicolumn{1}{l|}{} & \multicolumn{1}{l|}{} & \multicolumn{1}{l|}{} & \multicolumn{1}{l|}{} & \multicolumn{1}{l|}{} \\ \hline
1.00         &  Average                        &   1.4984     &     0.1000     &     1.4984     &   0.1000   \\  \hline
             &  Bias                           &   -1.56e-03  &    -2.28e-07   &    -1.56e-03   &  -2.54e-07 \\  \hline
             &  Variance                       &   1.27e-08   &     3.56e-13   &     1.72e-08   &   4.73e-13 \\  \hline
             &  Asym Var                       &   1.50e-08   &     3.52e-13   &     1.87e-08   &   4.39e-13 \\  \hline
\end{tabular}}
\label{table3}
\end{minipage}
\end{table}

\begin{table}[H]
\begin{minipage}{.5\linewidth}
\caption{LSEs when sample size is 300}
\centering
\resizebox{0.95\textwidth}{!}{\begin{tabular}{|c|c|c|c|c|c|}
\hline
\multicolumn{2}{|c|}{Type of error}            & \multicolumn{2}{c|}{N(0,$\sigma^2$)}          & \multicolumn{2}{c|}{MA(1) with $\rho = 0.5$}  \\ \hline
\multicolumn{2}{|c|}{Parameters}               & $\alpha$              & $\beta$               & $\alpha$              & $\beta$               \\ \hline
\multicolumn{2}{|c|}{True values}              & 1.5                   & 0.1                   & 1.5                   & 0.1                   \\ \hline
$\sigma^2$             & \multicolumn{1}{l|}{} & \multicolumn{4}{c|}{Least Squares Estimates}                                                  \\ \hline
0.10         &  Average                        &  1.4996      &    0.1000    &    1.4996      &      0.1000    \\  \hline
             &  Bias                           &  -4.10e-04   &   -1.48e-07  &   -4.12e-04    &     -1.50e-07  \\  \hline 
             &  Variance                       &  3.70e-10    &    4.60e-15  &    5.02e-10    &      5.26e-15  \\  \hline
             &  Asym Var                       &  4.44e-10    &    4.63e-15  &    5.56e-10    &      5.79e-15  \\  \hline
\multicolumn{1}{|l|}{} & \multicolumn{1}{l|}{} & \multicolumn{1}{l|}{} & \multicolumn{1}{l|}{} & \multicolumn{1}{l|}{} & \multicolumn{1}{l|}{} \\ \hline
0.50         &  Average                        &   1.4996     &   0.1000     &   1.4996      &      0.1000      \\  \hline
             &  Bias                           &   -4.08e-04  &  -1.47e-07   &  -4.11e-04    &     -1.46e-07    \\  \hline
             &  Variance                       &   1.79e-09   &   2.25e-14   &   2.64e-09    &      2.81e-14    \\  \hline
             &  Asym Var                       &   2.22e-09   &   2.31e-14   &   2.78e-09    &      2.89e-14    \\  \hline
\multicolumn{1}{|l|}{} & \multicolumn{1}{l|}{} & \multicolumn{1}{l|}{} & \multicolumn{1}{l|}{} & \multicolumn{1}{l|}{} & \multicolumn{1}{l|}{} \\ \hline
1.00         &  Average                        &  1.4996    &     0.1000     &   1.4996     &  0.1000        \\  \hline
             &  Bias                           &  -4.13e-04 &    -1.43e-07   &  -4.14e-04   & -1.61e-07      \\  \hline
             &  Variance                       &  4.04e-09  &     4.65e-14   &   5.34e-09   &  5.32e-14      \\  \hline
             &  Asym Var                       &  4.44e-09  &     4.63e-14   &   5.56e-09  &  5.79e-14       \\  \hline
\end{tabular}}
\label{table4}
\end{minipage}%
\begin{minipage}{.5\linewidth}
\caption{LSEs when sample size is 400}
\centering
\resizebox{0.95\textwidth}{!}{\begin{tabular}{|c|c|c|c|c|c|}
\hline
\multicolumn{2}{|c|}{Type of error}            & \multicolumn{2}{c|}{N(0,$\sigma^2$)}          & \multicolumn{2}{c|}{MA(1) with $\rho = 0.5$}  \\ \hline
\multicolumn{2}{|c|}{Parameters}               & $\alpha$              & $\beta$               & $\alpha$              & $\beta$               \\ \hline
\multicolumn{2}{|c|}{True values}              & 1.5                   & 0.1                   & 1.5                   & 0.1                   \\ \hline
$\sigma^2$             & \multicolumn{1}{l|}{} & \multicolumn{4}{c|}{Least Squares Estimates}                                                  \\ \hline
0.10         &  Average                        & 1.5000     &   0.1000     &     1.5000    &   0.1000         \\  \hline
             &  Bias                           & -2.39e-05  &   1.18e-07   &    -2.40e-05  &   1.18e-07       \\  \hline 
             &  Variance                       & 1.59e-10   &   1.17e-15   &     2.11e-10  &   1.37e-15       \\  \hline
             &  Asym Var                       & 1.88e-10   &   1.10e-15   &     2.34e-10  &   1.37e-15       \\  \hline
\multicolumn{1}{|l|}{} & \multicolumn{1}{l|}{} & \multicolumn{1}{l|}{} & \multicolumn{1}{l|}{} & \multicolumn{1}{l|}{} & \multicolumn{1}{l|}{} \\ \hline
0.50         &  Average                        &  1.5000     &   0.1000     &     1.5000    &   0.1000       \\  \hline
             &  Bias                           &  -2.26e-05  &   1.18e-07   &    -2.38e-05  &   1.20e-07     \\  \hline
             &  Variance                       &  7.91e-10   &   5.45e-15   &     1.01e-09  &   7.32e-15     \\  \hline
             &  Asym Var                       &  9.37e-10   &   5.49e-15   &     1.17e-09  &   6.87e-15     \\  \hline
\multicolumn{1}{|l|}{} & \multicolumn{1}{l|}{} & \multicolumn{1}{l|}{} & \multicolumn{1}{l|}{} & \multicolumn{1}{l|}{} & \multicolumn{1}{l|}{} \\ \hline
1.00         &  Average                        & 1.5000     &   0.1000      &    1.5000    &   0.1000    \\  \hline
             &  Bias                           & -2.60e-05  &   1.16e-07    &   -2.27e-05  &   1.17e-07  \\  \hline
             &  Variance                       & 1.55e-09   &   1.28e-14    &    2.14e-09  &   1.43e-14  \\  \hline
             &  Asym Var                       & 1.87e-09   &   1.10e-14    &    2.34e-09  &   1.37e-14  \\  \hline
\end{tabular}}
\label{table5}
\end{minipage}
\end{table}

\begin{table}[H]
\centering
\resizebox{0.475\textwidth}{!}{\begin{tabular}{|c|c|c|c|c|c|}
\hline
\multicolumn{2}{|c|}{Type of error}            & \multicolumn{2}{c|}{N(0,$\sigma^2$)}          & \multicolumn{2}{c|}{MA(1) with $\rho = 0.5$}  \\ \hline
\multicolumn{2}{|c|}{Parameters}               & $\alpha$              & $\beta$               & $\alpha$              & $\beta$               \\ \hline
\multicolumn{2}{|c|}{True values}              & 1.5                   & 0.1                   & 1.5                   & 0.1                   \\ \hline
$\sigma^2$             & \multicolumn{1}{l|}{} & \multicolumn{4}{c|}{Least Squares Estimates}                                                  \\ \hline
0.10         &  Average                        &   1.4999     &   0.1000    &    1.4999    &   0.1000         \\  \hline
             &  Bias                           &   -5.48e-05  &   7.01e-08  &   -5.43e-05  &   7.00e-08       \\  \hline 
             &  Variance                       &   9.18e-11   &   3.46e-16  &    1.08e-10  &   4.44e-16       \\  \hline
             &  Asym Var                       &   9.60e-11   &   3.60e-16  &    1.20e-10  &   4.50e-16       \\  \hline
\multicolumn{1}{|l|}{} & \multicolumn{1}{l|}{} & \multicolumn{1}{l|}{} & \multicolumn{1}{l|}{} & \multicolumn{1}{l|}{} & \multicolumn{1}{l|}{} \\ \hline
0.50         &  Average                        &   1.4999     &   0.1000    &    1.4999    &   0.1000       \\  \hline
             &  Bias                           &   -5.42e-05  &   7.04e-08  &   -5.33e-05  &   6.98e-08     \\  \hline
             &  Variance                       &   4.73e-10   &   1.64e-15  &    5.36e-10  &   2.28e-15     \\  \hline
             &  Asym Var                       &   4.80e-10   &   1.80e-15  &    6.00e-10  &   2.25e-15     \\  \hline
\multicolumn{1}{|l|}{} & \multicolumn{1}{l|}{} & \multicolumn{1}{l|}{} & \multicolumn{1}{l|}{} & \multicolumn{1}{l|}{} & \multicolumn{1}{l|}{} \\ \hline
1.00         &  Average                        &  1.4999     &   0.1000    &    1.4999    &   0.1000        \\  \hline
             &  Bias                           &  -5.65e-05  &   6.81e-08  &   -5.38e-05  &   7.33e-08      \\  \hline
             &  Variance                       &  9.17e-10   &   3.19e-15  &    1.14e-09  &   4.60e-15      \\  \hline
             &  Asym Var                       &  9.60e-10   &   3.60e-15  &    1.20e-09  &   4.50e-15      \\  \hline
\end{tabular}}
\caption{LSEs when sample size is 500}
\label{table6}
\end{table}

\subsection{Multiple component chirp-like model}\label{appendix:A2}

\begin{table}[H]
\begin{minipage}{.5\linewidth}
\caption{Estimates when n = 100}
\centering
\resizebox{0.95\textwidth}{!}{\begin{tabular}{|c|c|c|c|c|c|}
\hline
\multicolumn{2}{|c|}{Type of error} & \multicolumn{4}{c|}{N(0, $\sigma^2$)}           \\ \hline
\multicolumn{2}{|c|}{Parameters}    & $\alpha_1$ & $\beta_1$ & $\alpha_2$ & $\beta_2$ \\ \hline
\multicolumn{2}{|c|}{True values}   & 1.5        & 0.1       & 2.5        & 0.2       \\ \hline
$\sigma^2$       &                  &            &           &            &           \\ \hline
0.10             & Average          &  1.4986    &   0.1000   &  2.5068   &   0.2000     \\ \hline
                 & Bias             &  -1.42e-03 &  -3.65e-05 &  6.79e-03 &  -1.27e-05   \\ \hline
                 & MSE              &  1.33e-08  &   1.10e-12 &  2.80e-08 &   1.41e-12   \\ \hline
                 & Avar             &  1.20e-08  &   1.12e-12 &  1.88e-08 &   1.76e-12   \\ \hline
                 &                  &            &            &           &              \\ \hline
0.50             & Average          &  1.4986    &   0.1000   &  2.5068   &   0.2000     \\ \hline
                 & Bias             &  -1.42e-03 &  -3.65e-05 &  6.82e-03 &  -1.27e-05   \\ \hline
                 & MSE              &  6.92e-08  &   5.19e-12 &  1.36e-07 &   7.53e-12   \\ \hline
                 & Avar             &  6.00e-08  &   5.62e-12 &  9.38e-08 &   8.79e-12   \\ \hline
                 &                  &            &            &           &              \\ \hline
1.00             & Average          &  1.4986    &   0.1000   &  2.5068   &   0.2000     \\ \hline
                 & Bias             &  -1.39e-03 &  -3.64e-05 &  6.80e-03 &  -1.27e-05   \\ \hline
                 & MSE              &  1.57e-07  &   1.15e-11 &  2.74e-07 &   1.48e-11   \\ \hline
                 & Avar             &  1.20e-07  &   1.12e-11 &  1.88e-07 &   1.76e-11   \\ \hline
\multicolumn{6}{|l|}{}                                                                \\ \hline
\multicolumn{2}{|c|}{Type of error} & \multicolumn{4}{c|}{MA(1) with $\rho = 0.5$}    \\ \hline
Parameters       & Parameters       & $\alpha_1$ & $\beta_1$ & $\alpha_2$ & $\beta_2$ \\ \hline
True values      & True values      & 1.5        & 0.1       & 2.5        & 0.2       \\ \hline
$\sigma^2$       &                  &            &           &            &           \\ \hline
0.10             & Average          & 1.4986    &   0.1000   &  2.5068   &   0.2000     \\ \hline
                 & Bias             & -1.42e-03 &  -3.65e-05 &  6.80e-03 &  -1.27e-05   \\ \hline
                 & MSE              & 1.98e-08  &   1.61e-12 &  1.33e-08 &   1.85e-12   \\ \hline
                 & Avar             & 1.50e-08  &   1.41e-12 &  2.34e-08 &   2.20e-12   \\ \hline
                 &                  &           &            &           &              \\ \hline
0.50             & Average          & 1.4986    &   0.1000   &  2.5068   &   0.2000     \\ \hline
                 & Bias             & -1.44e-03 &  -3.65e-05 &  6.81e-03 &  -1.27e-05   \\ \hline
                 & MSE              & 9.86e-08  &   8.22e-12 &  7.24e-08 &   1.03e-11   \\ \hline
                 & Avar             & 7.50e-08  &   7.03e-12 &  1.17e-07 &   1.10e-11   \\ \hline
                 &                  &           &            &           &              \\ \hline
1.00             & Average          & 1.4986    &   0.1000   &  2.5068   &   0.2000     \\ \hline
                 & Bias             & -1.43e-03 &  -3.64e-05 &  6.81e-03 &  -1.26e-05   \\ \hline
                 & MSE              & 1.99e-07  &   1.52e-11 &  1.33e-07 &   1.98e-11   \\ \hline
                 & Avar             & 1.50e-07  &   1.41e-11 &  2.34e-07 &   2.20e-11   \\ \hline
\end{tabular}}
\label{table7}
\end{minipage}%
\begin{minipage}{.5\linewidth}
\caption{Estimates when n = 200}
\centering
\resizebox{0.95\textwidth}{!}{\begin{tabular}{|c|c|c|c|c|c|}
\hline
\multicolumn{2}{|c|}{Type of error} & \multicolumn{4}{c|}{N(0, $\sigma^2$)}           \\ \hline
\multicolumn{2}{|c|}{Parameters}    & $\alpha_1$ & $\beta_1$ & $\alpha_2$ & $\beta_2$ \\ \hline
\multicolumn{2}{|c|}{True values}   & 1.5        & 0.1       & 2.5        & 0.2       \\ \hline
$\sigma^2$       &                  &            &           &            &           \\ \hline
0.10             & Average          & 1.4984     &    0.1000    &   2.5025    &    0.2000           \\ \hline
                 & Bias             & -1.65e-03  &   -1.31e-06  &   2.46e-03  &   -3.87e-08         \\ \hline
                 & MSE              & 1.51e-09   &    6.16e-14  &   2.00e-09  &    4.96e-14         \\ \hline
                 & Avar             & 1.50e-09   &    3.52e-14  &   2.34e-09  &    5.49e-14         \\ \hline
                 &                  &            &              &             &                     \\ \hline
0.50             & Average          & 1.4984     &    0.1000    &   2.5025    &    0.2000           \\ \hline
                 & Bias             & -1.65e-03  &   -1.30e-06  &   2.46e-03  &   -4.71e-08         \\ \hline
                 & MSE              & 7.24e-09   &    3.07e-13  &   9.90e-09  &    2.60e-13         \\ \hline
                 & Avar             & 7.50e-09   &    1.76e-13  &   1.17e-08  &    2.75e-13         \\ \hline
                 &                  &            &              &             &                     \\ \hline
1.00             & Average          & 1.4984     &    0.1000    &   2.5025    &    0.2000           \\ \hline
                 & Bias             & -1.64e-03  &   -1.33e-06  &   2.45e-03  &   -5.00e-08         \\ \hline
                 & MSE              & 1.49e-08   &    5.84e-13  &   2.04e-08  &    4.94e-13         \\ \hline
                 & Avar             & 1.50e-08   &    3.52e-13  &   2.34e-08  &    5.49e-13         \\ \hline
\multicolumn{6}{|l|}{}                                                                \\ \hline
\multicolumn{2}{|c|}{Type of error} & \multicolumn{4}{c|}{MA(1) with $\rho = 0.5$}    \\ \hline
Parameters       & Parameters       & $\alpha_1$ & $\beta_1$ & $\alpha_2$ & $\beta_2$ \\ \hline
True values      & True values      & 1.5        & 0.1       & 2.5        & 0.2       \\ \hline
$\sigma^2$       &                  &            &           &            &           \\ \hline
0.10             & Average          & 1.4984     &    0.1000    &   2.5025    &    0.2000          \\ \hline
                 & Bias             & -1.65e-03  &   -1.31e-06  &   2.45e-03  &   -3.64e-08        \\ \hline
                 & MSE              & 2.09e-09   &    7.45e-14  &   9.07e-10  &    6.29e-14        \\ \hline
                 & Avar             & 1.87e-09   &    4.39e-14  &   2.93e-09  &    6.87e-14        \\ \hline
                 &                  &            &              &             &                    \\ \hline
0.50             & Average          & 1.4984     &    0.1000    &   2.5025    &    0.2000          \\ \hline
                 & Bias             & -1.64e-03  &   -1.32e-06  &   2.45e-03  &   -2.82e-08        \\ \hline
                 & MSE              & 1.03e-08   &    3.95e-13  &   4.83e-09  &    2.70e-13        \\ \hline
                 & Avar             & 9.37e-09   &    2.20e-13  &   1.46e-08  &    3.43e-13        \\ \hline
                 &                  &            &              &             &                    \\ \hline
1.00             & Average          & 1.4983     &    0.1000    &   2.5024    &    0.2000          \\ \hline
                 & Bias             & -1.65e-03  &   -1.31e-06  &   2.45e-03  &   -5.50e-08        \\ \hline
                 & MSE              & 2.00e-08   &    7.59e-13  &   1.03e-08  &    5.65e-13        \\ \hline
                 & Avar             & 1.87e-08   &    4.39e-13  &   2.93e-08  &    6.87e-13        \\ \hline
\end{tabular}}
\label{table8}
\end{minipage}
\end{table}

\begin{table}[H]
\begin{minipage}{.5\linewidth}
\caption{Estimates when n = 300}
\centering
\resizebox{0.95\textwidth}{!}{\begin{tabular}{|c|c|c|c|c|c|}
\hline
\multicolumn{2}{|c|}{Type of error} & \multicolumn{4}{c|}{N(0, $\sigma^2$)}           \\ \hline
\multicolumn{2}{|c|}{Parameters}    & $\alpha_1$ & $\beta_1$ & $\alpha_2$ & $\beta_2$ \\ \hline
\multicolumn{2}{|c|}{True values}   & 1.5        & 0.1       & 2.5        & 0.2       \\ \hline
$\sigma^2$       &                  &            &           &            &           \\ \hline
0.10             & Average          &  1.4999      &     0.1000     &    2.5001     &    0.2000              \\ \hline
                 & Bias             &  -1.10e-04   &    -1.61e-06   &    1.34e-04   &    8.47e-08            \\ \hline
                 & MSE              &  4.34e-10    &     4.10e-15   &    6.91e-10   &    6.32e-15            \\ \hline
                 & Avar             &  4.44e-10    &     4.63e-15   &    6.94e-10   &    7.23e-15            \\ \hline
                 &                  &              &                &               &                        \\ \hline
0.50             & Average          &  1.4999      &     0.1000     &    2.5001     &    0.2000              \\ \hline
                 & Bias             &  -1.11e-04   &    -1.60e-06   &    1.35e-04   &    8.55e-08            \\ \hline
                 & MSE              &  2.05e-09    &     2.18e-14   &    3.74e-09   &    3.21e-14            \\ \hline
                 & Avar             &  2.22e-09    &     2.31e-14   &    3.47e-09   &    3.62e-14            \\ \hline
                 &                  &              &                &               &                        \\ \hline
1.00             & Average          &  1.4999      &     0.1000     &    2.5001     &    0.2000              \\ \hline
                 & Bias             &  -1.07e-04   &    -1.60e-06   &    1.40e-04   &    8.60e-08            \\ \hline
                 & MSE              &  4.30e-09    &     4.55e-14   &    7.15e-09   &    6.73e-14            \\ \hline
                 & Avar             &  4.44e-09    &     4.63e-14   &    6.94e-09   &    7.23e-14            \\ \hline
\multicolumn{6}{|l|}{}                                                                \\ \hline
\multicolumn{2}{|c|}{Type of error} & \multicolumn{4}{c|}{MA(1) with $\rho = 0.5$}    \\ \hline
Parameters       & Parameters       & $\alpha_1$ & $\beta_1$ & $\alpha_2$ & $\beta_2$ \\ \hline
True values      & True values      & 1.5        & 0.1       & 2.5        & 0.2       \\ \hline
$\sigma^2$       &                  &            &           &            &           \\ \hline
0.10             & Average          &   1.4999      &     0.1000     &    2.5001     &    0.2000             \\ \hline
                 & Bias             &   -1.10e-04   &    -1.61e-06   &    1.35e-04   &    7.49e-08           \\ \hline
                 & MSE              &   5.67e-10    &     5.19e-15   &    3.32e-10   &    8.77e-15           \\ \hline
                 & Avar             &   5.56e-10    &     5.79e-15   &    8.68e-10   &    9.04e-15           \\ \hline
                 &                  &               &                &               &                       \\ \hline
0.50             & Average          &   1.4999      &     0.1000     &    2.5001     &    0.2000             \\ \hline
                 & Bias             &   -1.09e-04   &    -1.62e-06   &    1.34e-04   &    9.03e-08           \\ \hline
                 & MSE              &   2.86e-09    &     2.54e-14   &    1.71e-09   &    4.49e-14           \\ \hline
                 & Avar             &   2.78e-09    &     2.89e-14   &    4.34e-09   &    4.52e-14           \\ \hline
                 &                  &               &                &               &                       \\ \hline
1.00             & Average          &   1.4999      &     0.1000     &    2.5001     &    0.2000             \\ \hline
                 & Bias             &   -1.09e-04   &    -1.60e-06   &    1.35e-04   &    9.46e-08           \\ \hline
                 & MSE              &   5.57e-09    &     4.98e-14   &    3.46e-09   &    8.31e-14           \\ \hline
                 & Avar             &   5.56e-09    &     5.79e-14   &    8.68e-09   &    9.04e-14           \\ \hline
\end{tabular}}
\label{table9}
\end{minipage}%
\begin{minipage}{.5\linewidth}
\caption{Estimates when n = 400}
\centering
\resizebox{0.95\textwidth}{!}{\begin{tabular}{|c|c|c|c|c|c|}
\hline
\multicolumn{2}{|c|}{Type of error} & \multicolumn{4}{c|}{N(0, $\sigma^2$)}           \\ \hline
\multicolumn{2}{|c|}{Parameters}    & $\alpha_1$ & $\beta_1$ & $\alpha_2$ & $\beta_2$ \\ \hline
\multicolumn{2}{|c|}{True values}   & 1.5        & 0.1       & 2.5        & 0.2       \\ \hline
$\sigma^2$       &                  &            &           &            &           \\ \hline
0.10             & Average          &   1.5004     &   0.1000     &    2.4994     &   0.2000               \\ \hline
                 & Bias             &   3.54e-04   &   1.56e-07   &    -5.93e-04  &  -3.92e-08             \\ \hline
                 & MSE              &   1.66e-10   &   1.01e-15   &    2.87e-10   &   1.67e-15             \\ \hline
                 & Avar             &   1.88e-10   &   1.10e-15   &    2.93e-10   &   1.72e-15             \\ \hline
                 &                  &              &              &               &                        \\ \hline
0.50             & Average          &   1.5004     &   0.1000     &    2.4994     &   0.2000               \\ \hline
                 & Bias             &   3.53e-04   &   1.59e-07   &    -5.93e-04  &  -3.34e-08             \\ \hline
                 & MSE              &   8.61e-10   &   4.81e-15   &    1.52e-09   &   8.18e-15             \\ \hline
                 & Avar             &   9.37e-10   &   5.49e-15   &    1.46e-09   &   8.58e-15             \\ \hline
                 &                  &              &              &               &                        \\ \hline
1.00             & Average          &   1.5004     &   0.1000     &    2.4994     &   0.2000               \\ \hline
                 & Bias             &   3.53e-04   &   1.59e-07   &    -5.90e-04  &  -4.07e-08             \\ \hline
                 & MSE              &   1.69e-09   &   1.06e-14   &    3.06e-09   &   1.60e-14             \\ \hline
                 & Avar             &   1.87e-09   &   1.10e-14   &    2.93e-09   &   1.72e-14             \\ \hline
\multicolumn{6}{|l|}{}                                                                \\ \hline
\multicolumn{2}{|c|}{Type of error} & \multicolumn{4}{c|}{MA(1) with $\rho = 0.5$}    \\ \hline
Parameters       & Parameters       & $\alpha_1$ & $\beta_1$ & $\alpha_2$ & $\beta_2$ \\ \hline
True values      & True values      & 1.5        & 0.1       & 2.5        & 0.2       \\ \hline
$\sigma^2$       &                  &            &           &            &           \\ \hline
0.10             & Average          &   1.5004     &   0.1000     &    2.4994     &   0.2000              \\ \hline
                 & Bias             &   3.55e-04   &   1.55e-07   &    -5.93e-04  &  -3.86e-08            \\ \hline
                 & MSE              &   2.21e-10   &   1.30e-15   &    1.38e-10   &   2.13e-15            \\ \hline
                 & Avar             &   2.34e-10   &   1.37e-15   &    3.66e-10   &   2.15e-15            \\ \hline
                 &                  &              &              &               &                       \\ \hline
0.50             & Average          &   1.5004     &   0.1000     &    2.4994     &   0.2000              \\ \hline
                 & Bias             &   3.53e-04   &   1.57e-07   &    -5.93e-04  &  -3.55e-08            \\ \hline
                 & MSE              &   1.09e-09   &   5.73e-15   &    7.29e-10   &   1.05e-14            \\ \hline
                 & Avar             &   1.17e-09   &   6.87e-15   &    1.83e-09   &   1.07e-14            \\ \hline
                 &                  &              &              &               &                       \\ \hline
1.00             & Average          &   1.5003     &   0.1000     &    2.4994     &   0.2000              \\ \hline
                 & Bias             &   3.49e-04   &   1.59e-07   &    -5.95e-04  &  -4.41e-08            \\ \hline
                 & MSE              &   2.09e-09   &   1.25e-14   &    1.43e-09   &   2.12e-14            \\ \hline
                 & Avar             &   2.34e-09   &   1.37e-14   &    3.66e-09   &   2.15e-14            \\ \hline
\end{tabular}}
\label{table10}
\end{minipage}
\end{table}

\begin{table}[H]
\centering
\resizebox{0.475\textwidth}{!}{\begin{tabular}{|c|c|c|c|c|c|}
\hline
\multicolumn{2}{|c|}{Type of error} & \multicolumn{4}{c|}{N(0, $\sigma^2$)}           \\ \hline
\multicolumn{2}{|c|}{Parameters}    & $\alpha_1$ & $\beta_1$ & $\alpha_2$ & $\beta_2$ \\ \hline
\multicolumn{2}{|c|}{True values}   & 1.5        & 0.1       & 2.5        & 0.2       \\ \hline
$\sigma^2$       &                  &            &           &            &           \\ \hline
0.10             & Average          &   1.5000     &    0.1000    &   2.4998     &    0.2000             \\ \hline
                 & Bias             &   -1.86e-05  &   -4.15e-08  &   -2.27e-04  &   -3.72e-08           \\ \hline
                 & MSE              &   8.23e-11   &    3.95e-16  &   1.70e-10   &    5.36e-16           \\ \hline
                 & Avar             &   9.60e-11   &    3.60e-16  &   1.50e-10   &    5.62e-16           \\ \hline
                 &                  &              &              &              &                       \\ \hline
0.50             & Average          &   1.5000     &    0.1000    &   2.4998     &    0.2000             \\ \hline
                 & Bias             &   -1.96e-05  &   -4.08e-08  &   -2.26e-04  &   -3.37e-08           \\ \hline
                 & MSE              &   3.96e-10   &    1.79e-15  &   9.26e-10   &    2.69e-15           \\ \hline
                 & Avar             &   4.80e-10   &    1.80e-15  &   7.50e-10   &    2.81e-15           \\ \hline
                 &                  &              &              &              &                       \\ \hline
1.00             & Average          &   1.5000     &    0.1000    &   2.4998     &    0.2000             \\ \hline
                 & Bias             &   -1.70e-05  &   -3.90e-08  &   -2.30e-04  &   -3.47e-08           \\ \hline
                 & MSE              &   8.00e-10   &    3.69e-15  &   1.80e-09   &    6.13e-15           \\ \hline
                 & Avar             &   9.60e-10   &    3.60e-15  &   1.50e-09   &    5.62e-15           \\ \hline
\multicolumn{6}{|l|}{}                                                                \\ \hline
\multicolumn{2}{|c|}{Type of error} & \multicolumn{4}{c|}{MA(1) with $\rho = 0.5$}    \\ \hline
Parameters       & Parameters       & $\alpha_1$ & $\beta_1$ & $\alpha_2$ & $\beta_2$ \\ \hline
True values      & True values      & 1.5        & 0.1       & 2.5        & 0.2       \\ \hline
$\sigma^2$       &                  &            &           &            &           \\ \hline
0.10             & Average          &   1.5000     &    0.1000    &   2.4998     &    0.2000            \\ \hline
                 & Bias             &   -1.87e-05  &   -4.12e-08  &   -2.28e-04  &   -3.58e-08          \\ \hline
                 & MSE              &   1.14e-10   &    4.51e-16  &   7.81e-11   &    7.28e-16          \\ \hline
                 & Avar             &   1.20e-10   &    4.50e-16  &   1.87e-10   &    7.03e-16          \\ \hline
                 &                  &              &              &              &                      \\ \hline
0.50             & Average          &   1.5000     &    0.1000    &   2.4998     &    0.2000            \\ \hline
                 & Bias             &   -2.01e-05  &   -4.19e-08  &   -2.27e-04  &   -3.40e-08          \\ \hline
                 & MSE              &   5.42e-10   &    2.16e-15  &   3.99e-10   &    3.48e-15          \\ \hline
                 & Avar             &   6.00e-10   &    2.25e-15  &   9.37e-10   &    3.52e-15          \\ \hline
                 &                  &              &              &              &                      \\ \hline
1.00             & Average          &   1.5000     &    0.1000    &   2.4998     &    0.2000            \\ \hline
                 & Bias             &   -1.99e-05  &   -4.01e-08  &   -2.29e-04  &   -2.95e-08          \\ \hline
                 & MSE              &   1.05e-09   &    4.63e-15  &   7.70e-10   &    7.01e-15          \\ \hline
                 & Avar             &   1.20e-09   &    4.50e-15  &   1.87e-09   &    7.03e-15          \\ \hline
\end{tabular}}
\caption{Estimates when n = 500}
\label{table11}
\end{table}

\section{Some Preliminary Results}\label{appendix:B}
To provide the proofs of the asymptotic properties established in this manuscript, we will require following results: 

\begin{lemma}\label{preliminary_result_1}
If $\phi \in (0, \pi)$, then except for a countable number of points, the following hold true:
\begin{enumerate}[label=(\alph*)]
\item $\lim\limits_{n \rightarrow \infty} \frac{1}{n} \sum\limits_{t=1}^{n}\cos(\phi t) = \lim\limits_{n \rightarrow \infty} \frac{1}{n} \sum\limits_{t=1}^{n}\sin(\phi t) = 0.$
\item $\lim\limits_{n \rightarrow \infty} \frac{1}{n^{k+1}} \sum\limits_{t=1}^{n}t^{k} \cos^2(\phi t) = \lim\limits_{n \rightarrow \infty} \frac{1}{n^{k+1}} \sum\limits_{t=1}^{n}t^{k} \sin^2(\phi t) = \frac{1}{2(k+1)};\ k = 0, 1, 2, \cdots.$
\item $\lim\limits_{n \rightarrow \infty} \frac{1}{n^{k+1}} \sum\limits_{t=1}^{n}t^{k} \sin(\phi t) \cos(\phi t) = 0;\ k = 0, 1, 2, \cdots.$
\end{enumerate} 
\end{lemma}
\begin{proof}
Refer to Kundu and Nandi \cite{2012}. \\
\end{proof}

\begin{lemma}\label{preliminary_result_2}
If $\phi \in (0, \pi)$, then except for a countable number of points, the following hold true:
\begin{enumerate}[label=(\alph*)]
\item $\lim\limits_{n \rightarrow \infty} \frac{1}{n} \sum\limits_{t=1}^{n}\cos(\phi t^2) = \lim\limits_{n \rightarrow \infty} \frac{1}{n} \sum\limits_{t=1}^{n}\sin(\phi t^2) = 0.$
\item $\lim\limits_{n \rightarrow \infty} \frac{1}{n^{k+1}} \sum\limits_{t=1}^{n}t^{k} \cos^2(\phi t^2) = \lim\limits_{n \rightarrow \infty} \frac{1}{n^{k+1}} \sum\limits_{t=1}^{n}t^{k} \sin^2(\phi t^2) = \frac{1}{2(k+1)};\ k = 0, 1, 2, \cdots.$
\item $\lim\limits_{n \rightarrow \infty} \frac{1}{n^{k+1}} \sum\limits_{t=1}^{n}t^{k} \sin(\phi t^2) \cos(\phi t^2) = 0;\ k = 0, 1, 2, \cdots.$
\end{enumerate} 
\end{lemma}
\begin{proof}
Refer to Lahiri \cite{2011_2}.\\
\end{proof}

\begin{lemma}\label{preliminary_result_3}
If $(\phi_1, \phi_2) \in (0, \pi) \times (0, \pi)$, then except for a countable number of points, the following hold true:
\begin{enumerate}[label=(\alph*)]
\item $\lim\limits_{n \rightarrow \infty} \frac{1}{n^{k+1}} \sum\limits_{t=1}^{n} t^k \cos(\phi_1 t)\cos(\phi_2 t^2) = 0$
\item $\lim\limits_{n \rightarrow \infty} \frac{1}{n^{k+1}} \sum\limits_{t=1}^{n} t^k \cos(\phi_1 t)\sin(\phi_2 t^2) = 0$
\item $\lim\limits_{n \rightarrow \infty} \frac{1}{n^{k+1}} \sum\limits_{t=1}^{n} t^k \sin(\phi_1 t)\cos(\phi_2 t^2) = 0$
\item $\lim\limits_{n \rightarrow \infty} \frac{1}{n^{k+1}} \sum\limits_{t=1}^{n} t^k \sin(\phi_1 t)\sin(\phi_2 t^2) = 0$
\end{enumerate} 
$k = 0, 1, 2, \cdots$
\end{lemma}
\begin{proof} Consider the following:
\begin{flalign*}
&\frac{1}{n} \sum_{t=1}^{n}\cos(\phi_1 t)e^{i\phi_2 t^2} \leqslant \frac{1}{n} \{\sum_{t=1}^{n}\cos^2(\phi_1 t)\}^{1/2}\{\sum_{t=1}^{n}e^{2i\phi_2 t^2}\}^{1/2} = o(1) \textmd{ a.s. using Lemma 1.}&\\
& \textmd{Thus, we have, }\lim_{n \rightarrow \infty}\frac{1}{n} \sum_{t=1}^{n}\cos(\phi_1 t)cos(\phi_2 t^2) = 0 \textmd{ and } \lim_{n \rightarrow \infty}\frac{1}{n} \sum_{t=1}^{n}\cos(\phi_1 t)sin(\phi_2 t^2) = 0&
\end{flalign*}
Similarly, $\frac{1}{n} \sum\limits_{t=1}^{n}\sin(\phi_1 t)e^{i\phi_2 t^2} \xrightarrow{a.s.} 0$. Now the result can be generalised along the same lines as in proof of Lemma 2.2.1 of Lahiri \cite{2011_2}.\\

\end{proof}

\begin{lemma}\label{preliminary_result_4}
If X(t) satisfies assumptions \ref{assump:1}, \ref{assump:3} and \ref{assump:4}, then for $k \geqslant 0$:
\begin{enumerate}[label=(\alph*)]
\item $\sup\limits_{\phi} \bigg|\frac{1}{n^{k+1}} \sum\limits_{t=1}^{n} t^k X(t)e^{i(\phi t)}\bigg| \xrightarrow{a.s.} 0$
\item $\sup\limits_{\phi} \bigg|\frac{1}{n^{k+1}} \sum\limits_{t=1}^{n} t^k X(t)e^{i(\phi t^2)}\bigg| \xrightarrow{a.s.} 0$
\end{enumerate} 
Here $i = \sqrt{-i}.$
\end{lemma}
\begin{proof}
These can be obtained as particular cases of Lemma 2.2.2 of Lahiri \cite{2011_2}.\\
\end{proof}

\begin{lemma}\label{preliminary_result_5}
If $(\phi_1, \phi_2) \in (0, \pi) \times (0, \pi)$, then except for a countable number of points, the following hold true:
\begin{enumerate}[label=(\alph*)]
\item $\lim\limits_{n \rightarrow \infty} \frac{1}{n^k\sqrt{n}} \sum\limits_{t=1}^{n} t^k \cos(\phi_1 t^2) = \lim\limits_{n \rightarrow \infty} \frac{1}{n^k\sqrt{n}} \sum\limits_{t=1}^{n} t^k \sin(\phi_1 t^2)  = 0$
\item $\lim\limits_{n \rightarrow \infty} \frac{1}{n^k\sqrt{n}} \sum\limits_{t=1}^{n} t^k \cos(\phi_1 t)\cos(\phi_2 t) = 0$ 
\item $\lim\limits_{n \rightarrow \infty} \frac{1}{n^k\sqrt{n}} \sum\limits_{t=1}^{n} t^k \cos(\phi_1 t)\sin(\phi_2 t) = 0$
\item $\lim\limits_{n \rightarrow \infty} \frac{1}{n^k\sqrt{n}} \sum\limits_{t=1}^{n} t^k \sin(\phi_1 t)\sin(\phi_2 t) = 0$
\item $\lim\limits_{n \rightarrow \infty} \frac{1}{n^k\sqrt{n}} \sum\limits_{t=1}^{n} t^k \cos(\phi_1 t)\cos(\phi_2 t^2) = 0$
\item $\lim\limits_{n \rightarrow \infty} \frac{1}{n^k\sqrt{n}} \sum\limits_{t=1}^{n} t^k \cos(\phi_1 t)\sin(\phi_2 t^2) = 0$
\item $\lim\limits_{n \rightarrow \infty} \frac{1}{n^k \sqrt{n}} \sum\limits_{t=1}^{n} t^k \sin(\phi_1 t)\cos(\phi_2 t^2) = 0$
\item $\lim\limits_{n \rightarrow \infty}  \frac{1}{n^k \sqrt{n}} \sum\limits_{t=1}^{n} t^k \sin(\phi_1 t)\sin(\phi_2 t^2) = 0$
\end{enumerate} 
$k = 0, 1, 2, \cdots$
\end{lemma}
\begin{proof}
\textit{(a)} Consider the following:
\begin{flalign*}
& \sum\limits_{t=1}^{n}e^{i \phi_1 t^2} \leqslant [\sum\limits_{t=1}^{n}e^{2i\phi_1 t^2}]^{1/2} = [o(n)]^{1/2} \textmd{ using Lemma \ref{preliminary_result_2}\ (a).}& \\
&\textmd{Thus, we have } \frac{1}{\sqrt{n}}\cos(\phi_1 t^2) \rightarrow 0 \textmd{ as } n \rightarrow \infty \textmd{ and } \frac{1}{\sqrt{n}}\sin(\phi_1 t^2) \rightarrow 0 \textmd{ as } n \rightarrow \infty. &
\end{flalign*}
Now the result can be generalised for $k \geqslant 1$ along the same lines as in proof of Lemma 2.2.1 of Lahiri \cite{2011_2}.\\
The proofs of parts \textit{(b)}, \textit{(c)} and \textit{(d)} follow by using elementary product to sum trigonometric identities and then using the fact that: 
$$\sum_{t=1}^{n} e^{i\alpha t} = O(1); \ \alpha \in (-\pi,\pi).$$
The proofs of parts \textit{(e)-(h)}, follow similarly by using the above mentioned trigonometric identities and Lemma 5 of Grover et al. \cite{2018}.\\
\end{proof}

\section{One Component Chirp-like Model}\label{appendix:C}
\subsection{Proofs of the asymptotic properties of the LSEs}\label{appendix:C1}
We need the following lemmas to prove the consistency of the LSEs:
\begin{lemma}\label{lemma_condition_consistency}
Consider the set $S_c = \{\boldsymbol{\theta}: |\boldsymbol{\theta} - \boldsymbol{\theta}^0| > c; \boldsymbol{\theta} \in \boldsymbol{\Theta}\}$. If the following holds true:
\begin{equation}\label{condition_for_consistency}
\liminf \inf\limits_{S_c} \frac{1}{n} (Q(\boldsymbol{\theta}) - Q(\boldsymbol{\theta}^0)) > 0 \textmd{ a.s., }
\end{equation}
then $\hat{\boldsymbol{\theta}} \xrightarrow{a.s.} \boldsymbol{\theta}^0$ as $n \rightarrow \infty$ 
\end{lemma}
\begin{proof}
Let us denote $\hat{\boldsymbol{\theta}}$ by $\hat{\boldsymbol{\theta}}_{n}$, to highlight the fact that the estimates depend on the sample size $n$. Now suppose, $\hat{\boldsymbol{\theta}}_{n} \nrightarrow \boldsymbol{\theta}^0$, then there exists a subsequence $\{n_k\}$ of $\{n\}$, such that $\hat{\boldsymbol{\theta}}_{n_k} \nrightarrow \boldsymbol{\theta}^0$. In such a situation, one of two cases may arise:
\begin{enumerate}
\item $|\hat{A}_{n_k}| + |\hat{B}_{n_k}| + |\hat{C}_{n_k}| + |\hat{D}_{n_k}|$ is not bounded, that is, at least one of the $|\hat{A}_{n_k}|$ or $|\hat{B}_{n_k}|$ or $|\hat{C}_{n_k}|$ or $|\hat{D}_{n_k}|$ $\rightarrow \infty$ $\Rightarrow \frac{1}{n_k}Q_{n_k}(\hat{\boldsymbol{\theta}}_{n_k}) \rightarrow \infty$ \\
But, $\lim\limits_{n_k \rightarrow \infty} \frac{1}{n_k} Q_{n_k}(\boldsymbol{\theta}^0) < \infty$ which implies, $\frac{1}{n_k} (Q_{n_k}(\hat{\boldsymbol{\theta}}_{n_k}) - Q_{n_k}(\boldsymbol{\theta}^0)) \rightarrow \infty.$ This contradicts the fact that:
\begin{equation}\label{LSE_property}
Q_{n_k}(\hat{\boldsymbol{\theta}}_{n_k}) \leqslant Q_{n_k}(\boldsymbol{\theta}^0),
\end{equation}
which holds true as $\hat{\boldsymbol{\theta}}_{n_k}$ is the LSE of $\boldsymbol{\theta}^0$.
\item  $|\hat{A}_{n_k}| + |\hat{B}_{n_k}| + |\hat{C}_{n_k}| + |\hat{D}_{n_k}|$ is bounded, then there exists a $c > 0$ such that $\hat{\boldsymbol{\theta}}_{n_k} \in S_c$, for all $k = 1, 2, \cdots$. Now, since \eqref{condition_for_consistency} is true, this contradicts \eqref{LSE_property}.
\end{enumerate} 
 Hence, the result.\\
\end{proof}
\justify
\textit{Proof of Theorem \ref{consistency_LSEs_one_component}:}
Consider the difference:
\begin{flalign*}
&\frac{1}{n}\bigg(Q(\boldsymbol{\theta}) - Q(\boldsymbol{\theta}^0)\bigg)&\\
& = \frac{1}{n}\sum_{t=1}^{n}\bigg(y(t) - A\cos(\alpha t) - B \sin(\alpha t) - C \cos(\beta t^2) - D \sin(\beta t^2)\bigg)^2& \\
& - \frac{1}{n}\sum_{t=1}^{n}\bigg(y(t) - A^0\cos(\alpha^0 t) - B^0 \sin(\alpha^0 t) - C^0 \cos(\beta^0 t^2) - D^0 \sin(\beta^0 t^2)\bigg)^2 & \\
& = \frac{1}{n} \sum_{t=1}^{n}\bigg(A^0 \cos(\alpha^0 t) - A \cos(\alpha t) + B^0 \sin(\alpha^0 t) - B \sin(\alpha t) + C^0 \cos(\beta^0 t^2) - C \cos(\beta t^2) & \\
& + D^0 \sin(\beta^0 t^2) - D \sin(\beta t^2)\bigg)^2 + \frac{1}{n} \sum_{t=1}^{n} X(t) \bigg(A^0 \cos(\alpha^0 t) - A \cos(\alpha t) + B^0 \sin(\alpha^0 t) - B \sin(\alpha t) & \\
& + C^0 \cos(\beta^0 t^2) - C \cos(\beta t^2) + D^0 \sin(\beta^0 t^2) - D \sin(\beta t^2)\bigg) &\\
& = f(\boldsymbol{\theta}) + g(\boldsymbol{\theta}).&\\
\end{flalign*}
Now using Lemma \ref{preliminary_result_4}, it can be easily seen that:
\begin{equation}\label{lim_g_theta}
\lim_{n \rightarrow \infty} \sup_{\boldsymbol{\theta} \in S_c} g(\boldsymbol{\theta}) = 0 \textmd{ a.s. }
\end{equation}
Thus, we have:
\begin{equation*}
\liminf \inf_{\boldsymbol{\theta} \in S_c} \frac{1}{n}\bigg(Q(\boldsymbol{\theta}) - Q(\boldsymbol{\theta}^0)\bigg) = \liminf \inf_{\boldsymbol{\theta} \in S_c} f(\boldsymbol{\theta}).
\end{equation*} 
Note that the proof will follow if we show that $ \liminf \inf_{\boldsymbol{\theta} \in S_c} f(\boldsymbol{\theta}) > 0$. Consider the set $S_c = \{\boldsymbol{\theta}: |\boldsymbol{\theta} - \boldsymbol{\theta}^0| \geqslant 6c; \boldsymbol{\theta} \in \boldsymbol{\Theta}\} \subset S_c^{(1)} \cup S_c^{(2)} \cup S_c^{(3)} \cup S_c^{(4)} \cup S_c^{(5)} \cup S_c^{(6)}$, where 
\begin{flalign*}
& S_c^{(1)} = \{\boldsymbol{\theta}: |A - A^0| \geqslant c; \boldsymbol{\theta} \in \boldsymbol{\Theta}\} \qquad S_c^{(2)} = \{\boldsymbol{\theta}: |B - B^0| \geqslant c; \boldsymbol{\theta} \in \boldsymbol{\Theta}\} &\\
& S_c^{(3)} = \{\boldsymbol{\theta}: |\alpha - \alpha^0| \geqslant c; \boldsymbol{\theta} \in \boldsymbol{\Theta}\} \qquad S_c^{(4)} = \{\boldsymbol{\theta}: |C - C^0| \geqslant c; \boldsymbol{\theta} \in \boldsymbol{\Theta}\} & \\
& S_c^{(5)} = \{\boldsymbol{\theta}: |D - D^0| \geqslant c; \boldsymbol{\theta} \in \boldsymbol{\Theta}\} \qquad S_c^{(6)} = \{\boldsymbol{\theta}: |\beta - \beta^0| \geqslant c; \boldsymbol{\theta} \in \boldsymbol{\Theta}\} & \\
\end{flalign*} 
Now, we split the set $S_c^{(1)}$ as follows:
\begin{flalign*}
&S_c^{(1)} = \{\boldsymbol{\theta}: |A - A^0| \geqslant c; \boldsymbol{\theta} \in \boldsymbol{\Theta}\}& \\ 
& \subset \{\boldsymbol{\theta}: |A - A^0| \geqslant c; \boldsymbol{\theta} \in \boldsymbol{\Theta}; \alpha = \alpha^0; \beta = \beta^0 \} \cup \{\boldsymbol{\theta}: |A - A^0| \geqslant c; \boldsymbol{\theta} \in \boldsymbol{\Theta}; \alpha \neq \alpha^0; \beta = \beta^0 \}  & \\
& \cup \{\boldsymbol{\theta}: |A - A^0| \geqslant c; \boldsymbol{\theta} \in \boldsymbol{\Theta}; \alpha = \alpha^0; \beta \neq \beta^0 \} \cup \{\boldsymbol{\theta}: |A - A^0| \geqslant c; \boldsymbol{\theta} \in \boldsymbol{\Theta}; \alpha \neq \alpha^0; \beta \neq \beta^0 \} &\\ 
& = S_c^{(1)_{1}} \cup S_c^{(1)_{2}} \cup S_c^{(1)_{3}} \cup S_c^{(1)_{4}} &
\end{flalign*}
Now let us consider:
\begin{flalign*}
&\liminf \inf_{\boldsymbol{\theta} \in S_c^{(1)_{1}}} f(\boldsymbol{\theta})& \\
& = \liminf \inf_{\boldsymbol{\theta} \in S_c^{(1)_{1}}} \frac{1}{n} \sum_{t=1}^{n} \bigg(A^0 \cos(\alpha^0 t) - A \cos(\alpha t) + B^0 \sin(\alpha^0 t) - B \sin(\alpha t) + C^0 \cos(\beta^0 t^2) - C \cos(\beta t^2) & \\
& + D^0 \sin(\beta^0 t^2) - D \sin(\beta t^2)\bigg)^2& \\
& = \liminf \inf_{\boldsymbol{\theta} \in S_c^{(1)_{1}}} \frac{1}{n} \sum_{t=1}^{n} \bigg((A^0 - A) \cos(\alpha^0 t) + (B^0 - B) \sin(\alpha^0 t) + (C^0 - C) \cos(\beta^0 t^2) + (D^0 - D) \sin(\beta^0 t^2)\bigg)^2& \\
& = \frac{(A^0 - A)^2}{2} + \frac{(B^0 - B)^2}{2} + \frac{(C^0 - C)^2}{2} + \frac{(D^0 - D)^2}{2} > 0 &
\end{flalign*}
\begin{flalign*}
&\liminf \inf_{\boldsymbol{\theta} \in S_c^{(1)_{2}}} f(\boldsymbol{\theta})& \\
& = \liminf \inf_{\boldsymbol{\theta} \in S_c^{(1)_{1}}} \frac{1}{n} \sum_{t=1}^{n} \bigg(A^0 \cos(\alpha^0 t) - A \cos(\alpha t) + B^0 \sin(\alpha^0 t) - B \sin(\alpha t) + (C^0  - C) \cos(\beta^0 t^2) & \\
& \qquad + (D^0 - D) \sin(\beta^0 t^2)\bigg)^2 =  \frac{{A^0}^2}{2} + \frac{{A}^2}{2} + \frac{{B^0}^2}{2} + \frac{{B}^2}{2} + \frac{(C^0 - C)^2}{2} + \frac{(D^0 - D)^2}{2} > 0 &
\end{flalign*}
\begin{flalign*}
&\liminf \inf_{\boldsymbol{\theta} \in S_c^{(1)_{3}}} f(\boldsymbol{\theta})& \\
& = \liminf \inf_{\boldsymbol{\theta} \in S_c^{(1)_{3}}} \frac{1}{n} \sum_{t=1}^{n} \bigg((A^0 - A) \cos(\alpha^0 t) + (B^0 - B) \sin(\alpha^0 t) + C^0 \cos(\beta^0 t^2) - C \cos(\beta t^2) & \\
& \qquad + D^0 \sin(\beta^0 t^2) - D \sin(\beta t^2)\bigg)^2 = \frac{(A^0 - A)^2}{2} + \frac{(B^0 - B)^2}{2} + \frac{{C^0}^2}{2} + \frac{C^2}{2} + \frac{{D^0}^2}{2} + \frac{{D}^2}{2} > 0 &
\end{flalign*}
\begin{flalign*}
& \textmd{Finally, } \liminf \inf_{\boldsymbol{\theta} \in S_c^{(1)_{4}}} f(\boldsymbol{\theta})& \\
& = \liminf \inf_{\boldsymbol{\theta} \in S_c^{(1)_{1}}} \frac{1}{n} \sum_{t=1}^{n} \bigg(A^0 \cos(\alpha^0 t) - A \cos(\alpha t) + B^0 \sin(\alpha^0 t) - B \sin(\alpha t) + C^0 \cos(\beta^0 t^2) - C \cos(\beta t^2) & \\
& \qquad + D^0 \sin(\beta^0 t^2) - D \sin(\beta t^2)\bigg)^2 =  \frac{{A^0}^2}{2} + \frac{{A}^2}{2} + \frac{{B^0}^2}{2} + \frac{{B}^2}{2} + \frac{{C^0}^2}{2} + \frac{C^2}{2} + \frac{{D^0}^2}{2} + \frac{{D}^2}{2} > 0 &
\end{flalign*}
Note that we used lemmas \ref{preliminary_result_1} and \ref{preliminary_result_2} in all the above computations of the limits. On combining all the above, we have $ \liminf \inf\limits_{\boldsymbol{\theta} \in S_c^{(1)}} f(\boldsymbol{\theta}) > 0.$ Similarly, it can be shown that the result holds for the rest of the sets. Therefore, by Lemma \ref{lemma_condition_consistency}, $\hat{\boldsymbol{\theta}}$ is a strongly consistent estimator of $\boldsymbol{\theta}^0$.\\ \qed
\justify
\textit{Proof of Theorem \ref{asymptotic_dist_LSEs_one_component}:} 
To obtain the asymptotic distribution of the LSEs, we express $\textbf{Q}'(\hat{\boldsymbol{\theta}})$ using multivariate Taylor series expansion arount the point $\boldsymbol{\theta}^0$, as follows:
\begin{equation}\label{taylor_series_Q'}
\textbf{Q}'(\hat{\boldsymbol{\theta}}) - \textbf{Q}'(\boldsymbol{\theta}^0) = (\hat{\boldsymbol{\theta}} - \boldsymbol{\theta}^0) \textbf{Q}''(\bar{\boldsymbol{\theta}}).
\end{equation}
Here, $\bar{\boldsymbol{\theta}}$ is a point between $\hat{\boldsymbol{\theta}}$ and $\boldsymbol{\theta}^0$. 
Since, $\hat{\boldsymbol{\theta}}$ is the LSE of $\boldsymbol{\theta}^0$, $\textbf{Q}'(\hat{\boldsymbol{\theta}}) = 0$. Thus, we have:
\begin{equation}\label{taylor_series_Q'_2}
(\hat{\boldsymbol{\theta}} - \boldsymbol{\theta}^0) = - \textbf{Q}'(\boldsymbol{\theta}^0)[\textbf{Q}''(\bar{\boldsymbol{\theta}})]^{-1}.
\end{equation}
Multiplying both sides of  \eqref{taylor_series_Q'_2}) by the $6 \times 6$ diagonal matrix $\textbf{D} = diag(\frac{1}{\sqrt{n}}, \frac{1}{\sqrt{n}}, \frac{1}{n\sqrt{n}}, \frac{1}{\sqrt{n}}, \frac{1}{\sqrt{n}}, \frac{1}{n^2\sqrt{n}})$, we get:
\begin{equation}\label{taylor_series_Q'_3}
(\hat{\boldsymbol{\theta}} - \boldsymbol{\theta}^0)\textbf{D}^{-1} = - \textbf{Q}'(\boldsymbol{\theta}^0)\textbf{D}[\textbf{D}\textbf{Q}''(\bar{\boldsymbol{\theta}})\textbf{D}]^{-1}.
\end{equation}
First, we will show that:
 \begin{equation}\label{Q'D_limit_dist}
\lim_{n \rightarrow \infty} \textbf{Q}'(\boldsymbol{\theta}^0)\textbf{D} \xrightarrow{d} N(0, 4c \sigma^2 \boldsymbol{\Sigma}).
\end{equation}
Here, 
\begin{equation}\label{Sigma_matrix}
\boldsymbol{\Sigma} = \begin{pmatrix}
\frac{1}{2} & 0 & \frac{B^0}{4} & 0 & 0 & 0 \\
0 & \frac{1}{2} & \frac{-A^0}{4} & 0 & 0 & 0 \\
\frac{B^0}{4} & \frac{-A^0}{4} & \frac{{A^0}^2 + {B^0}^2}{6} & 0 & 0 & 0 \\
0 & 0 & 0 & \frac{1}{2} & 0 & \frac{D^0}{6} \\
0 & 0 & 0 & 0 & \frac{1}{2} & \frac{-C^0}{6} \\
0 & 0 & 0 &  \frac{D^0}{6} &  \frac{-C^0}{6} &  \frac{{C^0}^2 + {D^0}^2}{10}
\end{pmatrix}
\end{equation}
\justify
To prove \eqref{Q'D_limit_dist}), we compute the elements of the $6 \times 1$ vector \\
$\textbf{Q}'(\boldsymbol{\theta}^0)\textbf{D} = \begin{pmatrix}
\frac{1}{\sqrt{n}}\frac{\partial Q(\boldsymbol{\theta})}{\partial A} & \frac{1}{\sqrt{n}}\frac{\partial Q(\boldsymbol{\theta})}{\partial B} & \frac{1}{n\sqrt{n}}\frac{\partial Q(\boldsymbol{\theta})}{\partial \alpha} & \frac{1}{\sqrt{n}}\frac{\partial Q(\boldsymbol{\theta})}{\partial C} & \frac{1}{\sqrt{n}}\frac{\partial Q(\boldsymbol{\theta})}{\partial D} & \frac{1}{n^2\sqrt{n}}\frac{\partial Q(\boldsymbol{\theta})}{\partial \beta}
\end{pmatrix}$ as follows:
\begin{flalign*}
& \frac{1}{\sqrt{n}}\frac{\partial Q(\boldsymbol{\theta})}{\partial A}  = \frac{-2}{\sqrt{n}} \sum_{t=1}^{n} \bigg(y(t) - A \cos(\alpha t) - B \sin(\alpha t) - C \cos(\beta t^2) - D \sin(\beta t^2)\bigg)\cos(\alpha t)&\\
& \Rightarrow  \frac{1}{\sqrt{n}}\frac{\partial Q(\boldsymbol{\theta}^0)}{\partial A} =  \frac{-2}{\sqrt{n}} \sum_{t=1}^{n} X(t)\cos(\alpha^0 t).&
\end{flalign*}
Similarly, the rest of the elements can be computed and we get:  
\begin{equation*}
\textbf{Q}'(\boldsymbol{\theta}^0)\textbf{D} = \begin{pmatrix}
\frac{-2}{\sqrt{n}} \sum\limits_{t=1}^{n} X(t)\cos(\alpha^0 t) \\
\frac{-2}{\sqrt{n}} \sum\limits_{t=1}^{n} X(t)\sin(\alpha^0 t)\\
\frac{-2}{n\sqrt{n}} \sum\limits_{t=1}^{n} t X(t)\big(-A^0\sin(\alpha^0 t) + B^0 \cos(\alpha^0 t)\big)\\
\frac{-2}{\sqrt{n}}\sum\limits_{t=1}^{n} X(t)\cos(\beta^0 t^2)\\
\frac{-2}{\sqrt{n}}\sum\limits_{t=1}^{n} X(t)\sin(\beta^0 t^2)\\
\frac{-2}{n^2\sqrt{n}} \sum\limits_{t=1}^{n} t^2 X(t)\big(-C^0\sin(\beta^0 t^2) + D^0 \cos(\beta^0 t^2)\big)
\end{pmatrix}.
\end{equation*}
Now using the Central Limit Theorem (CLT) of stochastic processes (see Fuller \cite{2009}, the above vector tends to a 6-variate Gaussian distribution with mean 0 and variance $4 c \sigma^2 \boldsymbol{\Sigma}$ and hence \eqref{Q'D_limit_dist} holds true.
Now we consider the second derivative matrix $\textbf{D}\textbf{Q}''(\bar{\boldsymbol{\theta}})\textbf{D}$. Note that, since $\hat{\boldsymbol{\theta}} \xrightarrow{a.s.} \boldsymbol{\theta}^0$ as $n \rightarrow \infty$ and $\bar{\boldsymbol{\theta}}$ is a point between $\hat{\boldsymbol{\theta}}$ and $\boldsymbol{\theta}^0$, 
$$\lim\limits_{n \rightarrow \infty} \textbf{D}\textbf{Q}''(\bar{\boldsymbol{\theta}})\textbf{D} = \lim\limits_{n \rightarrow \infty} \textbf{D}\textbf{Q}''(\boldsymbol{\theta}^0)\textbf{D}.$$  
Using lemmas \ref{preliminary_result_1}, \ref{preliminary_result_2}, \ref{preliminary_result_3} and \ref{preliminary_result_4} and after some calculations, it can be shown that:
\begin{equation}\label{DQ''D_limit}
\textbf{D}\textbf{Q}''(\boldsymbol{\theta}^0)\textbf{D} = 2\boldsymbol{\Sigma},
\end{equation}
where $\boldsymbol{\Sigma}$ is as defined in  \eqref{Sigma_matrix}. On combining, \eqref{taylor_series_Q'_3},\eqref{Q'D_limit_dist} and \eqref{DQ''D_limit}, the desired result follows.\\ \qed
\subsection{Proofs of the asymptotic properties of the sequential LSEs}\label{appendix:C2}
Following lemmas are required to prove the consistency of the sequential LSEs:
\begin{lemma}\label{lemma_condition_consistency_sinusoid_1}
Let us define the set $M_c = \{\boldsymbol{\theta}^{(1)}: |\boldsymbol{\theta}^{(1)} - {\boldsymbol{\theta}^0}^{(1)}| \geqslant3 c; \boldsymbol{\theta}^{(1)} \in \boldsymbol{\Theta}^{(1)}\}$. If the following holds true:
\begin{equation}\label{condition_for_consistency_sinusoid_1}
\liminf \inf\limits_{M_c} \frac{1}{n} (Q_1(\boldsymbol{\theta}^{(1)}) - Q_1({\boldsymbol{\theta}^0}^{(1)})) > 0 \textmd{ a.s. }
\end{equation}
then $\tilde{\boldsymbol{\theta}}^{(1)} \xrightarrow{a.s.} {\boldsymbol{\theta}^0}^{(1)}$ as $n \rightarrow \infty$ 
\end{lemma}
\begin{proof}
This can be proved by contradiction along the same lines as Lemma \ref{lemma_condition_consistency}.\\
\end{proof}

\begin{lemma}\label{lemma_condition_consistency_chirp_1}
Let us define the set $N_c = \{\boldsymbol{\theta}^{(2)} : \boldsymbol{\theta}^{(2)} \in \boldsymbol{\Theta}^{(2)} ;\ |\boldsymbol{\theta}^{(2)} - {\boldsymbol{\theta}^0}^{(2)}| \geqslant 3c\}.$ If for any $c>0$, 
\begin{equation}\label{condition_for_consistency_chirp_1}
\liminf \inf\limits_{\boldsymbol{\theta}^{(2)} \in N_c} \frac{1}{n} (Q_2(\boldsymbol{\theta}^{(2)}) - Q_2({\boldsymbol{\theta}^0}^{(2)})) > 0 \textmd{ a.s. }
\end{equation}
then $\tilde{\boldsymbol{\theta}}^{(2)} \xrightarrow{a.s.} {\boldsymbol{\theta}^0}^{(2)}$ as $n \rightarrow \infty.$  
\end{lemma}
\begin{proof}
This can be proved by contradiction along the same lines as Lemma \ref{lemma_condition_consistency}.\\
\end{proof}

\justify
\textit{Proof of Theorem \ref{consistency_sequential_LSEs_one_comp}:}
First we prove the consistency of the parameter estimates of the sinusoid component, $\tilde{\boldsymbol{\theta}}^{(1)}$. For this, consider the difference:
\begin{flalign*}
&\frac{1}{n}(Q_1(\boldsymbol{\theta}^{(1)}) - Q_1({\boldsymbol{\theta}^0}^{(1)}))&\\
& = \frac{1}{n}\Bigg[\sum_{t=1}^{n}\bigg(y(t) - A\cos(\alpha t) - B\sin(\alpha t)\bigg)^2 - \bigg(y(t) - A^0 \cos(\alpha^0 t) - B^0 \sin(\alpha^0 t)\bigg)^2 \Bigg] &\\
&= \frac{1}{n} \sum_{t=1}^n\bigg(A^0 \cos(\alpha^0 t)  - A\cos(\alpha t) +  B^0 \sin(\alpha^0 t) - B \sin(\alpha t) + C^0 \cos(\beta^0 t^2) + D^0 \sin(\beta^0 t^2) + X(t)\bigg)^2&\\
&\quad - \frac{1}{n}\sum_{t=1}^{n}\bigg(C^0 \cos(\beta^0 t^2) + D^0 \sin(\beta^0 t^2) + X(t)\bigg)^2& \\
& =  \frac{1}{n} \sum_{t=1}^n\bigg(A^0 \cos(\alpha^0 t) +  B^0 \sin(\alpha^0 t) - A\cos(\alpha t) - B \sin(\alpha t) \bigg)^2&\\
& \quad + \frac{2}{n}\sum_{t=1}^n(C^0 \cos(\beta^0 t^2)+ D^0 \sin(\beta^0 t^2) + X(t) \bigg)\bigg(A^0 \cos(\alpha^0 t) +  B^0 \sin(\alpha^0 t) - A\cos(\alpha t) - B \sin(\alpha t)\bigg)&\\
&= f_1(\boldsymbol{\theta}^{(1)}) + g_1(\boldsymbol{\theta}^{(1)}).&
\end{flalign*} 
Here, $$f_1(\boldsymbol{\theta}^{(1)}) =  \frac{1}{n} \sum_{t=1}^n\bigg(A^0 \cos(\alpha^0 t) +  B^0 \sin(\alpha^0 t) - A\cos(\alpha t) - B \sin(\alpha t) \bigg)^2 \textmd{ and, }$$ $$g_1(\boldsymbol{\theta}^{(1)}) = \frac{2}{n} \sum_{t=1}^{n}\bigg(C^0 \cos(\beta^0 t^2)+ D^0 \sin(\beta^0 t^2) + X(t) \bigg)\bigg(A^0 \cos(\alpha^0 t) +  B^0 \sin(\alpha^0 t) - A\cos(\alpha t) - B \sin(\alpha t)\bigg)$$
Now using lemmas \ref{preliminary_result_3} and \ref{preliminary_result_4}, it is easy to see that:
\begin{equation*}
\sup\limits_{\boldsymbol{\theta} \in M_c} |g_1(\boldsymbol{\theta}^{(1)})| \xrightarrow{a.s.} 0.
\end{equation*}
Thus if we prove that $\liminf \inf\limits_{M_c}f_1(\boldsymbol{\theta}^{(1)}) > 0$ a.s., it will follow that $\liminf \inf\limits_{M_c} \frac{1}{n} (Q_1(\boldsymbol{\theta}^{(1)}) - Q_1({\boldsymbol{\theta}^0}^{(1)})) > 0 \textmd{ a.s. }$. 
First consider the set $M_c = \{\boldsymbol{\theta}^{(1)}: |\boldsymbol{\theta}^{(1)} - {\boldsymbol{\theta}^0}^{(1)}| \geqslant 3c; \boldsymbol{\theta}^{(1)} \in \boldsymbol{\Theta}^{(1)}\}$. It is evident that: $$M_c \subset M_c^{(1)} \cup M_c^{(2)} \cup M_c^{(3)},$$ where $M_c^{(1)} = \{\boldsymbol{\theta}^{(1)}: |A - A^0| \geqslant c; \boldsymbol{\theta}^{(1)} \in \boldsymbol{\Theta}^{(1)}\}$, $M_c^{(2)} = \{\boldsymbol{\theta}^{(1)}: |B - B^0| \geqslant c; \boldsymbol{\theta}^{(1)} \in \boldsymbol{\Theta}^{(1)}\}$ and $M_c^{(3)} = \{\boldsymbol{\theta}^{(1)}: |\alpha - \alpha^0| \geqslant c; \boldsymbol{\theta}^{(1)} \in \boldsymbol{\Theta}^{(1)}\}.$
Now we further split the set $M_c^{(1)}$ which can be written as: $M_c^{(1)_{1}} \cup M_c^{(1)_{2}}$, where 
\begin{flalign*}
& M_c^{(1)_{1}} = \{\boldsymbol{\theta}^{(1)}: |A - A^0| \geqslant c; \boldsymbol{\theta}^{(1)} \in \boldsymbol{\Theta}^{(1)}; \alpha = \alpha^0\} \textmd{ and } M_c^{(1)_{2}} = \{\boldsymbol{\theta}^{(1)}: |A - A^0| \geqslant c; \boldsymbol{\theta}^{(1)} \in \boldsymbol{\Theta}^{(1)}; \alpha \neq \alpha^0\}&\\
& \textmd{Consider, } \liminf \inf\limits_{M_c^{(1)_{1}}}f_1(\boldsymbol{\theta}^{(1)}) = \liminf \inf\limits_{M_c^{(1)_{1}}}  \frac{1}{n} \sum_{t=1}^n\bigg(A^0 \cos(\alpha^0 t) +  B^0 \sin(\alpha^0 t) - A\cos(\alpha t) - B \sin(\alpha t) \bigg)^2& \\
& = \frac{(A^0 - A)^2}{2} + \frac{(B^0 - B)^2}{2} > 0\textmd{ a.s.} \textmd{ (using Lemma \ref{preliminary_result_1}).}&\\
&\textmd{Again, using Lemma \ref{preliminary_result_1}, } \liminf \inf\limits_{M_c^{(1)_{2}}}  \frac{1}{n} \sum_{t=1}^n\bigg(A^0 \cos(\alpha^0 t) +  B^0 \sin(\alpha^0 t) - A\cos(\alpha t) - B \sin(\alpha t) \bigg)^2 & \\
&  \qquad \qquad \qquad \qquad  \qquad \qquad \qquad =\  \frac{{A^0}^2}{2} + \frac{{B^0}^2}{2} + \frac{A^2}{2} + \frac{B^2}{2} > 0 \textmd{  a.s.} .&
\end{flalign*}
Similarly, it can be shown that $\liminf \inf\limits_{M_c^{(2)}}f_1(\boldsymbol{\theta}^{(1)}) > 0$ a.s. and $\liminf \inf\limits_{M_c^{(3)}}f_1(\boldsymbol{\theta}^{(1)}) > 0$ a.s.. Now using Lemma \ref{lemma_condition_consistency_sinusoid_1}, $\tilde{A}$, $\tilde{B}$ and $\tilde{\alpha}$ are strongly consistent estimators of $A^0$, $B^0$ and $\alpha^0$ respectively. To prove the consistency of the chirp parameter sequential estimates, $\tilde{C}$, $\tilde{D}$ and $\tilde{\beta}$, we need the following lemma:
\begin{lemma}\label{rate_of_convergence_1}
If assumptions 1,2 and 3 are satisfied, then:
$$(\tilde{\boldsymbol{\theta}}^{(1)} - {\boldsymbol{\theta}^0}^{(1)})(\sqrt{n}\textbf{D}_1)^{-1} \xrightarrow{a.s.} 0.$$
Here, $\textbf{D}_1 = diag(\frac{1}{\sqrt{n}}, \frac{1}{\sqrt{n}}, \frac{1}{n\sqrt{n}})$.
\end{lemma} 
\begin{proof}
Consider the error sum of squares:
$Q_1(\boldsymbol{\theta}) = \frac{1}{n}\sum\limits_{t=1}^n(y(t) - A \cos(\alpha t) - B \sin(\alpha t))^2.$\\
By Taylor series expansion of $\textbf{Q}_1'(\tilde{\boldsymbol{\theta}}^{(1)})$ around the point ${\boldsymbol{\theta}^0}^{(1)}$, we get:
\begin{equation}\label{Q'_1_taylor}
\textbf{Q}_1'(\tilde{\boldsymbol{\theta}}^{(1)}) - \textbf{Q}_1'({\boldsymbol{\theta}^0}^{(1)}) = (\tilde{\boldsymbol{\theta}}^{(1)} - {\boldsymbol{\theta}^0}^{(1)})\textbf{Q}_1''(\bar{\boldsymbol{\theta}}^{(1)})
\end{equation}
where, $\bar{\boldsymbol{\theta}}^{(1)}$ is a point lying between $\tilde{\boldsymbol{\theta}}^{(1)}$ and ${\boldsymbol{\theta}^0}^{(1)}$.
Since, $\tilde{\boldsymbol{\theta}}^{(1)}$ minimises $Q_1(\boldsymbol{\theta})$, it implies that $\textbf{Q}_1'(\tilde{\boldsymbol{\theta}}^{(1)}) = 0$ and therefore \eqref{Q'_1_taylor} can be written as:
\begin{flalign}
\label{taylor_series_Q1} &(\tilde{\boldsymbol{\theta}}^{(1)} - {\boldsymbol{\theta}^0}^{(1)}) = - \textbf{Q}_1'({\boldsymbol{\theta}^0}^{(1)})[\textbf{Q}_1''(\bar{\boldsymbol{\theta}}^{(1)})]^{-1}&\\ 
\label{taylor_series_Q1_D} & \Rightarrow (\tilde{\boldsymbol{\theta}}^{(1)} -{\boldsymbol{\theta}^0}^{(1)})(\sqrt{n}\textbf{D}_1)^{-1} = [- \frac{1}{\sqrt{n}}\textbf{Q}_1'({\boldsymbol{\theta}^0}^{(1)})\textbf{D}_1][\textbf{D}_1 \textbf{Q}_1''(\bar{\boldsymbol{\theta}}^{(1)})\textbf{D}_1]^{-1}&
\end{flalign}
Now let us calculate the right hand side explicitly. First consider the first derivative vector $\frac{1}{\sqrt{n}}\textbf{Q}_1'({\boldsymbol{\theta}^0}^{(1)})\textbf{D}_1$.
$$\frac{1}{\sqrt{n}}\textbf{Q}_1'({\boldsymbol{\theta}^0}^{(1)})\textbf{D}_1 = \begin{pmatrix}
\frac{1}{n}\frac{\partial Q_1({\boldsymbol{\theta}^0}^{(1)})}{\partial A} & \frac{1}{n}\frac{\partial Q_1({\boldsymbol{\theta}^0}^{(1)})}{\partial B} & \frac{1}{n^2}\frac{\partial Q_1({\boldsymbol{\theta}^0}^{(1)})}{\partial \alpha}
\end{pmatrix}$$
By straight forward calculations and using lemmas \ref{preliminary_result_3} and \ref{preliminary_result_4}\textit{(a)}, one can easily see that:
\begin{equation}\label{Q1'_1_result}
\frac{1}{\sqrt{n}}\textbf{Q}_1'({\boldsymbol{\theta}^0}^{(1)})\textbf{D}_1 \rightarrow 0 \textmd{ a.s. }
\end{equation}
Now let us consider the second derivative matrix $\textbf{D}_1 \textbf{Q}_1''(\bar{\boldsymbol{\theta}}^{(1)})\textbf{D}_1$. Since $\tilde{\boldsymbol{\theta}}^{(1)} \xrightarrow{a.s.} {\boldsymbol{\theta}^0}^{(1)}$ and $\bar{\boldsymbol{\theta}}^{(1)}$ is a point between them, we have:
$$\textbf{D}_1 \textbf{Q}_1''(\bar{\boldsymbol{\theta}}^{(1)})\textbf{D}_1 = \lim_{n \rightarrow \infty} \textbf{D}_1 \textbf{Q}_1''({\boldsymbol{\theta}^0}^{(1)})\textbf{D}_1$$
Again by routine calculations and using lemmas \ref{preliminary_result_1}, \ref{preliminary_result_3} and \ref{preliminary_result_4}\textit{(a)} , one can evaluate each element of this $3 \times 3$ matrix, and get: 
\begin{equation}\label{Q1''_1_result}
\lim_{n \rightarrow \infty} \textbf{D}_1 \textbf{Q}_1''({\boldsymbol{\theta}^0}^{(1)})\textbf{D}_1 = 2 \boldsymbol{\Sigma}_1,
\end{equation}
where $\boldsymbol{\Sigma}_1  = \begin{pmatrix}
\frac{1}{2} & 0 & \frac{B^0}{4} \\
0 & \frac{1}{2} & \frac{-A^0}{4} \\
\frac{B^0}{4} & \frac{-A^0}{4} & \frac{{A^0}^2 + {B^0}^2}{6} \\
\end{pmatrix} > 0,$ a positive definite matrix. Hence combining \eqref{Q1'_1_result} and \eqref{Q1''_1_result}, we get the desired result.\\
\end{proof}
\justify
Using the above lemma, we get the following relationship between the sinusoid component of the model and its estimate:
\begin{equation}\label{relation_sinusoid_comp_and_estimate}
\tilde{A} \cos(\tilde{\alpha} t) + \tilde{B} \sin(\tilde{\alpha} t) = A^0 \cos(\alpha^0 t) + B^0 \sin(\alpha^0 t) + o(1)
\end{equation}
Now to prove the consistency of $\tilde{\boldsymbol{\theta}}^{(1)} = (\tilde{C}, \tilde{D}, \tilde{\beta})$, we consider the following difference:
\begin{flalign*}
&\frac{1}{n}(Q_2({\boldsymbol{\theta}}^{(2)}) - Q_2({\boldsymbol{\theta}^0}^{(2)}))&\\
& = \frac{1}{n}\Bigg[\sum_{t=1}^{n}\bigg(y_1(t) - C\cos(\beta t^2) - D\sin(\beta t^2)\bigg)^2 - \bigg(y_1(t) - C^0 \cos(\beta^0 t^2) - D^0 \sin(\beta^0 t^2)\bigg)^2 \Bigg] &\\
& =  \frac{1}{n} \sum_{t=1}^n\bigg(C^0 \cos(\beta^0 t^2) +  D^0 \sin(\beta^0 t^2) - C\cos(\beta t^2) - D \sin(\beta t^2) \bigg)^2&\\
& \quad + \frac{2}{n}\sum_{t=1}^n(A^0 \cos(\alpha^0 t)+ B^0 \sin(\alpha^0 t^2) + X(t) \bigg)\bigg(C^0 \cos(\beta^0 t^2) +  D^0 \sin(\beta^0 t^2) - C\cos(\beta t^2) - D \sin(\beta t^2)\bigg)&\\
&= f_2({\boldsymbol{\theta}}^{(2)}) + g_2({\boldsymbol{\theta}}^{(2)}). &
\end{flalign*}
Using lemmas \ref{preliminary_result_3} and \ref{preliminary_result_4}, we have 
\begin{equation*}
\sup\limits_{{\boldsymbol{\theta}} \in N_c} |g_2({\boldsymbol{\theta}}^{(2)})| \xrightarrow{a.s.} 0,
\end{equation*}
and using straight forward, but lengthy calculations and splitting the set $N_c$, similar to the splitting of set $M_c$, before, it can be shown that $\liminf\inf\limits_{\boldsymbol{\xi} \in N_c} f_2({\boldsymbol{\theta}}^{(2)}) > 0$.\\
$\textmd{Thus, }\tilde{\boldsymbol{\theta}}^{(2)} \xrightarrow{a.s.} {\boldsymbol{\theta}^0}^{(2)} \textmd{ as } n \rightarrow \infty \textmd{ by Lemma \ref{lemma_condition_consistency_chirp_1}.}$ Hence, the result. \\ \qed
\justify
\textit{Proof of Theorem \ref{asymptotic_dist_sequential_LSEs_one_comp}:} We first examine the asymptotic distribution of the sequential estimates of the sinusoid component, that is $\tilde{\boldsymbol{\theta}}^{(1)}$ From \ref{taylor_series_Q1}, we have:
$$(\tilde{\boldsymbol{\theta}}^{(1)} - {\boldsymbol{\theta}^0}^{(1)})\textbf{D}_1^{-1} = - \textbf{Q}_1'({\boldsymbol{\theta}^0}^{(1)})\textbf{D}_1[\textbf{D}_1\textbf{Q}_1''(\bar{\boldsymbol{\theta}}^{(1)})\textbf{D}_1]^{-1}.$$
First we show that $\textbf{Q}_1'({\boldsymbol{\theta}^0}^{(1)})\textbf{D}_1 \rightarrow N_3(0, 4 \sigma^2 c \boldsymbol{\Sigma}_1).$ We compute the elements of the derivative vector $\textbf{Q}_1'({\boldsymbol{\theta}^0}^{(1)})$ and using Lemma \ref{preliminary_result_5} \textit{(e), (f), (g)} and \textit{(h)}, we obtain:
\begin{equation}\label{Q1_D1}
\textbf{Q}_1'({\boldsymbol{\theta}^0}^{(1)})\textbf{D}_1 \overset{a.eq.}{=} -2 \begin{pmatrix}
\frac{1}{\sqrt{n}}\sum\limits_{t=1}^{n} X(t) \cos(\alpha^0 t) \\
\frac{1}{\sqrt{n}}\sum\limits_{t=1}^{n} X(t) \sin(\alpha^0 t) \\
\frac{1}{n\sqrt{n}}\sum\limits_{t=1}^{n} t X(t)(-A_1^0\sin(\alpha^0 t) + B^0 \cos(\alpha^0 t)) 
\end{pmatrix}.
\end{equation}
Here, $\overset{a.eq.}{=}$ means asymptotically equivalent. Now again using CLT, the right hand side of \eqref{Q1_D1} tends to 3-variate Gaussian distribution with mean
0 and variance-covariance matrix, $4 \sigma^2 c \boldsymbol{\Sigma}_1.$ Using this and \eqref{Q1''_1_result}, we have the desired result.
\justify
Next we determine the asymptotic distribution of $\tilde{\boldsymbol{\theta}}^{(2)}.$ For this, we consider the error sum of squares, $Q_2(\boldsymbol{\theta}^{(2)})$ as defined in  \eqref{ess_2}. Let $\boldsymbol{Q}'_2({\boldsymbol{\theta}}^{(2)})$ be the first derivative vector and $\boldsymbol{Q}''_2(\boldsymbol{\theta}^{(2)})$, the second derivative matrix of $Q_2(\boldsymbol{\theta}^{(2)})$.Using multivariate  Taylor series expansion, we expand $\boldsymbol{Q}'_2(\tilde{\boldsymbol{\theta}}^{(2)})$ around the point ${\boldsymbol{\theta}^0}^{(2)}$, and get:
$$(\tilde{\boldsymbol{\theta}}^{(2)} - {\boldsymbol{\theta}^0}^{(2)}) = -\boldsymbol{Q}'_2({\boldsymbol{\theta}^0}^{(2)}) [\boldsymbol{Q}''_2(\bar{\boldsymbol{\theta}}^{(2)})]^{-1}.$$
Multiplying both sides by the matrix $\boldsymbol{D}_2^{-1}$, where $\boldsymbol{D}_2 = diag(\frac{1}{\sqrt{n}}, \frac{1}{\sqrt{n}},\frac{1}{n^2\sqrt{n}})$, we get:
$$(\tilde{\boldsymbol{\theta}}^{(2)} - {\boldsymbol{\theta}^0}^{(2)})\textbf{D}_2^{-1} = -\boldsymbol{Q}'_2({\boldsymbol{\theta}^0}^{(2)})\textbf{D}_2 [\textbf{D}_2\boldsymbol{Q}''_2(\bar{\boldsymbol{\theta}}^{(2)})\textbf{D}_2]^{-1}.$$ Now when we evaluate the first derivative vector $\boldsymbol{Q}'_2({\boldsymbol{\theta}^0}^{(2)})\textbf{D}_2$, we obtain (using Lemma \ref{preliminary_result_5} \textit{(a)}):
\begin{equation}\label{Q2'_D2}
\boldsymbol{Q}'_2({\boldsymbol{\theta}^0}^{(2)})\textbf{D}_2 \overset{a.eq.}{=} -2 \begin{pmatrix}
\frac{1}{\sqrt{n}}\sum\limits_{t=1}^{n} X(t) \cos(\beta^0 t^2) \\
\frac{1}{\sqrt{n}}\sum\limits_{t=1}^{n} X(t) \sin(\beta^0 t^2) \\
\frac{1}{n^2\sqrt{n}}\sum\limits_{t=1}^{n} t X(t)(-C^0\sin(\beta^0 t^2) + D^0 \cos(\beta^0 t^2)) 
\end{pmatrix}.
\end{equation}
Again using the CLT, the vector on the right hand side of  \eqref{Q2'_D2} tends to $N_3(0, 4\sigma^2 c \boldsymbol{\Sigma}_2),$ where
 $\boldsymbol{\Sigma}_2  = \begin{pmatrix}
\frac{1}{2} & 0 & \frac{D^0}{6} \\
0 & \frac{1}{2} & \frac{-C^0}{6} \\
\frac{D^0}{6} & \frac{-C^0}{6} & \frac{{C^0}^2 + {D^0}^2}{10} \\
\end{pmatrix} > 0.$ \\
Note that:
$$\lim_{n \rightarrow \infty}\textbf{D}_2\boldsymbol{Q}''_2(\bar{\boldsymbol{\theta}}^{(2)})\textbf{D}_2 = \lim_{n \rightarrow \infty}\textbf{D}_2\boldsymbol{Q}''_2({\boldsymbol{\theta}^0}^{(2)})\textbf{D}_2.$$
On computing the second derivative $3 \times 3$ matrix $\textbf{D}_2\boldsymbol{Q}''_2({\boldsymbol{\theta}^0}^{(2)})\textbf{D}_2$ and using lemmas \ref{preliminary_result_2}, \ref{preliminary_result_3} and \ref{preliminary_result_4} \textit{(b)}, we get:
\begin{equation}\label{D2_Q2''_D2}
\lim_{n \rightarrow \infty}\textbf{D}_2\boldsymbol{Q}''_2({\boldsymbol{\theta}^0}^{(2)})\textbf{D}_2 = 2\boldsymbol{\Sigma}_2.
\end{equation}
Combining results \eqref{Q2'_D2} and \eqref{D2_Q2''_D2}, we get the stated asymptotic distribution of $\tilde{\boldsymbol{\theta}}^{(2)}.$
 Hence, the result.\\ \qed
 
\section{Multiple Component Chirp-like model}\label{appendix:D} 
\subsection{Proofs of the asymptotic properties of the LSEs}\label{appendix:D1}
\textit{Proof of Theorm \ref{asymptotic_distribution_LSEs_multiple_comp}:} Consider the error sum of squares, defined in  \eqref{ess_multiple_component}. Let us denote $\textbf{Q}'(\boldsymbol{\vartheta})$ as the $3(p+q) \times 1$ first derivative vector and $\textbf{Q}''(\boldsymbol{\vartheta})$ as the  $3(p+q) \times 3(p+q)$ second derivative matrix.
Using multivariate Taylor series expansion, we have:
$$\textbf{Q}'(\hat{\boldsymbol{\vartheta}}) - \textbf{Q}'(\boldsymbol{\vartheta}^0) = (\hat{\boldsymbol{\vartheta}} - \boldsymbol{\vartheta}^0)\textbf{Q}''(\bar{\boldsymbol{\vartheta}}). $$ 
Here $\bar{\boldsymbol{\vartheta}}$ is a point between $\hat{\boldsymbol{\vartheta}}$ and $\boldsymbol{\vartheta}^0.$ Now using the fact that $\textbf{Q}'(\hat{\boldsymbol{\vartheta}}) = 0$ and multiplying both sides of the above equation by $\mathfrak{D}^{-1}$, we have:
$$(\hat{\boldsymbol{\vartheta}} - \boldsymbol{\vartheta}^0)\mathfrak{D}^{-1} = - \textbf{Q}'(\hat{\boldsymbol{\vartheta}})\mathfrak{D}[\mathfrak{D}\textbf{Q}''(\bar{\boldsymbol{\vartheta}})\mathfrak{D}]^{-1}.$$
Also note that, $(\hat{\boldsymbol{\vartheta}} - \boldsymbol{\vartheta}^0)\mathfrak{D}^{-1} = \bigg(({\hat{\boldsymbol{\theta}}_1}^{(1)} - {\boldsymbol{\theta}_1^0}^{(1)}), \cdots, ({\hat{\boldsymbol{\theta}}_p}^{(1)} - {\boldsymbol{\theta}_p^0}^{(1)}), ({\hat{\boldsymbol{\theta}}_{1}}^{(2)} - {\boldsymbol{\theta}_{1}^0}^{(2)}), \cdots, ({\hat{\boldsymbol{\theta}}_q}^{(2)} - {\boldsymbol{\theta}_q^0}^{(2)})  \bigg)\mathfrak{D}^{-1}.$\\
\justify
Now we evaluate the elements of the vector $\textbf{Q}'(\boldsymbol{\vartheta}^0)$ and the matrix $\textbf{Q}''(\bar{\boldsymbol{\vartheta}})$:
\begin{flalign*}
&\frac{\partial Q(\boldsymbol{\vartheta})}{\partial A_j}\bigg|_{\boldsymbol{\vartheta}^0} = -2 \sum_{t=1}^{n} X(t)\cos(\alpha_j^0 t),  \quad  \frac{\partial Q(\boldsymbol{\vartheta})}{\partial B_j}\bigg|_{\boldsymbol{\vartheta}^0} = -2 \sum_{t=1}^{n} X(t)\sin(\alpha_j^0 t), \textmd{ and} &\\
&\qquad \frac{\partial Q(\boldsymbol{\vartheta})}{\partial \alpha_j}\bigg|_{\boldsymbol{\vartheta}^0} = -2 \sum_{t=1}^{n} t X(t)\bigg(-A_j^0 \sin(\alpha_j^0 t) + B_j^0 \cos(\alpha_j^0 t)\bigg), \textmd{ for } j = 1, \cdots, p. & \\
& \textmd{Similarly, for } k = 1, \cdots, q,\ \frac{\partial Q(\boldsymbol{\vartheta})}{\partial C_k}\bigg|_{\boldsymbol{\vartheta}^0} = -2 \sum_{t=1}^{n} X(t)\cos(\beta_k^0 t^2),  \quad  \frac{\partial Q(\boldsymbol{\vartheta})}{\partial D_k}\bigg|_{\boldsymbol{\vartheta}^0} = -2 \sum_{t=1}^{n} X(t)\sin(\beta_k^0 t^2) \textmd{ and} &\\
&\qquad \frac{\partial Q(\boldsymbol{\vartheta})}{\partial \beta_k}\bigg|_{\boldsymbol{\vartheta}^0} = -2 \sum_{t=1}^{n} t^2 X(t)\bigg(-C_k^0 \sin(\beta_k^0 t^2) + D_k^0 \cos(\beta_k^0 t)\bigg).&
\end{flalign*}
\begin{flalign*}
&\frac{\partial^2 Q(\boldsymbol{\vartheta})}{\partial A_j^2}\bigg|_{\boldsymbol{\vartheta}^0}  = 2\sum_{t=1}^{n}\cos^2(\alpha_j^0 t), \ \frac{\partial^2 Q(\boldsymbol{\vartheta})}{\partial B_j^2}\bigg|_{\boldsymbol{\vartheta}^0}  = 2\sum_{t=1}^{n}\sin^2(\alpha_j^0 t),\ j = 1, \cdots, p,& \\ 
& \qquad \qquad \qquad  \frac{\partial^2 Q(\boldsymbol{\vartheta})}{\partial C_k^2}\bigg|_{\boldsymbol{\vartheta}^0}  = 2\sum_{t=1}^{n}\cos^2(\beta_k^0 t^2) \textmd{ and}  \frac{\partial^2 Q(\boldsymbol{\vartheta})}{\partial D_k^2}\bigg|_{\boldsymbol{\vartheta}^0}  = 2\sum_{t=1}^{n}\sin^2(\beta_k^0 t^2),\ k = 1, \cdots, q. & \\
& \frac{\partial^2 Q(\boldsymbol{\vartheta})}{\partial A_j \partial B_j}\bigg|_{\boldsymbol{\vartheta}^0} = 2 \sum_{t=1}^{n} \sin(\alpha_j^0 t) \cos(\alpha_j^0 t),&\\
& \frac{\partial^2 Q(\boldsymbol{\vartheta})}{\partial A_j \partial \alpha_j}\bigg|_{\boldsymbol{\vartheta}^0} = 2 \sum_{t=1}^{n}t X(t) \sin(\alpha_j^0 t) - 2 A_j^0 \sum_{t=1}^{n} t \cos(\alpha_j^0 t) \sin(\alpha_j^0 t) + 2 B_j^0 \sum_{t=1}^{n} t\cos^2(\alpha_j^0 t), &\\
 & \frac{\partial^2 Q(\boldsymbol{\vartheta})}{\partial A_j \partial C_k}\bigg|_{\boldsymbol{\vartheta}^0} = 2 \sum_{t=1}^{n} \cos(\beta_k^0 t^2) \cos(\alpha_j^0 t),\ \frac{\partial^2 Q(\boldsymbol{\vartheta})}{\partial A_j \partial D_k}\bigg|_{\boldsymbol{\vartheta}^0} = 2 \sum_{t=1}^{n} \sin(\beta_k^0 t^2) \cos(\alpha_j^0 t), &\\
& \frac{\partial^2 Q(\boldsymbol{\vartheta})}{\partial A_j \partial \beta_k}\bigg|_{\boldsymbol{\vartheta}^0} = - 2 C_k^0 \sum_{t=1}^{n} t^2 \cos(\alpha_j^0 t) \sin(\beta_k^0 t^2) + 2 D_k^0 \sum_{t=1}^{n} t^2\cos(\alpha_j^0 t) \cos(\beta_k^0 t^2). &
\end{flalign*}
Similarly the rest of the partial derivatives can be computed and using lemmas \ref{preliminary_result_1}, \ref{preliminary_result_2}, \ref{preliminary_result_3} and \ref{preliminary_result_4}, it can be shown that:
$$\mathfrak{D}\textbf{Q}''(\bar{\boldsymbol{\vartheta}})\mathfrak{D} \rightarrow 2 \mathcal{E}(\boldsymbol{\vartheta}^0).$$
Now, using CLT on the first derivative vector, $\textbf{Q}'(\boldsymbol{\vartheta}^0)\mathfrak{D}$, it can be shown that it converges to a multivariate Gaussian distribution. Using routine calculations, and again using lemmas \ref{preliminary_result_1}, \ref{preliminary_result_2}, \ref{preliminary_result_3} and \ref{preliminary_result_4}, we compute the asymptotic variances for each of the elements and their covariances and we get:
$$\textbf{Q}'(\boldsymbol{\vartheta}^0)\mathfrak{D} \xrightarrow{d} N_{3(p+q)}(0, 4c \sigma^2 \mathcal{E}(\boldsymbol{\vartheta}^0)).$$
Hence, the result.\\ \qed

\subsection{Proofs of the asymptotic properties of the LSEs}\label{appendix:D2} 
To prove theorems \ref{consistency_first component} and \ref{consistency_rest_of_the_components}, we need the following lemmas:

\begin{lemma}\label{lemma_condition_consistency_multiple}
\begin{enumerate}[label=(\alph*)]
\item Consider the set $M_c^{(j)} = \{\boldsymbol{\theta}^{(1)}_j: |\boldsymbol{\theta}^{(1)}_j - {\boldsymbol{\theta}_j^0}^{(1)}| \geqslant3 c; \boldsymbol{\theta}_j^{(1)} \in \boldsymbol{\Theta}^{(1)}\},\ j = 1, \cdots, p$. If the following holds true:
\begin{equation}\label{condition_for_consistency_sinusoid}
\liminf \inf\limits_{M_c^{(j)}} \frac{1}{n} (Q_{2j-1}(\boldsymbol{\theta}^{(1)}_j) - Q_{2j-1}({\boldsymbol{\theta}_j^0}^{(1)})) > 0 \textmd{ a.s. }
\end{equation}
then $\tilde{\boldsymbol{\theta}}_j^{(1)} \xrightarrow{a.s.} {\boldsymbol{\theta}_j^0}^{(1)}$ as $n \rightarrow \infty$ \\
\item Let us define the set $N_c^{(k)} = \{\boldsymbol{\theta}_k^{(2)} : \boldsymbol{\theta}_k^{(2)} \in \boldsymbol{\Theta}^{(2)} ;\ |\boldsymbol{\theta}_k^{(2)} - {\boldsymbol{\theta}_k^0}^{(2)}| \geqslant 3c\},\ k = 1, \cdots, q.$ If for any $c>0$, 
\begin{equation}\label{condition_for_consistency_chirp}
\liminf \inf\limits_{\boldsymbol{\theta}_k^{(2)} \in N_c^{(k)}} \frac{1}{n} (Q_{2k}(\boldsymbol{\theta}_k^{(2)}) - Q_{2k}({\boldsymbol{\theta}_k^0}^{(2)})) > 0 \textmd{ a.s. }
\end{equation}
then $\tilde{\boldsymbol{\theta}}_k^{(2)} \xrightarrow{a.s.} {\boldsymbol{\theta}_k^0}^{(2)}$ as $n \rightarrow \infty.$  
\end{enumerate}
\end{lemma}
\begin{proof}
This can be proved by contradiction along the same lines as Lemma \ref{lemma_condition_consistency}.
\end{proof}

\begin{lemma}\label{rate_of_convergence_all}
If the assumptions 1, 3 and 4 are satisfied, then for $j \leqslant p$ and $k \leqslant q$:
\begin{enumerate}[label=(\alph*)]
\item $(\tilde{\boldsymbol{\theta}_j} - \boldsymbol{\theta}_j^0)(\sqrt{n}\boldsymbol{D}_1)^{-1} \xrightarrow{a.s.} 0.$
\item $(\tilde{\boldsymbol{\xi}_k} - \boldsymbol{\xi}_k^0)(\sqrt{n}\boldsymbol{D}_2)^{-1} \xrightarrow{a.s.} 0.$
\end{enumerate}
Here, $\boldsymbol{D}_1 = diag(\frac{1}{\sqrt{n}}, \frac{1}{\sqrt{n}}, \frac{1}{n\sqrt{n}})$ and $\boldsymbol{D}_2 = diag(\frac{1}{\sqrt{n}}, \frac{1}{\sqrt{n}}, \frac{1}{n^2\sqrt{n}})$.
\end{lemma}
\begin{proof}
This proof can be obtained along the same lines  as Lemma \ref{rate_of_convergence_1}.
\end{proof}
\justify
Now the proofs of theorems  \ref{consistency_first component} and \ref{consistency_rest_of_the_components} can be obtained by using the above lemmas and following the same argument as in Theorem \ref{consistency_sequential_LSEs_one_comp}.\\
\justify
Next we examine the situation when the number of components are over estimated (see Theorem \ref{consistency_excess_components}). The proof of Theorem \ref{consistency_excess_components} will follow consequently from the below stated lemmas:\\

\begin{lemma}\label{over_estimation_sin_lses}
If $X(t)$, is the error component as defined before, and if $\tilde{A}$, $\tilde{B}$ and $\tilde{\alpha}$ are obtained by minimizing the following function:
$$Q_{p+q+1}(\boldsymbol{\theta}^{(1)}) = \frac{1}{n}\sum_{t=1}^{N}\bigg(X(t) - A \cos(\alpha t) - B \sin(\alpha t)\bigg)^2,$$
then $\tilde{A} \xrightarrow{a.s.} 0$ and $\tilde{B} \xrightarrow{a.s.} 0.$ 
\end{lemma}
\begin{proof}
The sum of squares function $Q_{p+q+1}(\boldsymbol{\theta}^{(1)})$ can be written as:
\begin{flalign*}
&\frac{1}{n} \sum_{t=1}^{n}X^2(t) - \frac{2}{n} \sum_{t=1}^{n} X(t)\bigg(A \cos(\alpha t) + B \sin(\alpha t)\bigg) + \frac{A^2 + B^2}{2} + o(1)&\\
= & R(\boldsymbol{\theta}^{(1)}) + o(1).&
\end{flalign*}
Since the difference between $Q_{p+q+1}(\boldsymbol{\theta}^{(1)})$ and $R(\boldsymbol{\theta}^{(1)})$ is $o(1)$, replacing former with latter will have negligible effect on the estimators. 
Thus, we have
\begin{equation*}
\tilde{A} = \frac{2}{n}\sum\limits_{t=1}^{n} X(t)\cos(\alpha t) + o(1)\textmd{ and } \tilde{B} = \frac{2}{n}\sum_{t=1}^{n} X(t)\sin(\alpha t) + o(1).
\end{equation*}
Now using Lemma \ref{preliminary_result_4} \textit{(a)}, the result follows. 
\end{proof}
\begin{lemma}\label{over_estimation_chirp_lses}
If $X(t)$, is the error component as defined before, and if $\tilde{C}$, $\tilde{D}$ and $\tilde{\beta}$ are obtained by minimizing the following function:
$$\frac{1}{n}\sum_{t=1}^{N}\bigg(X(t) - C \cos(\beta t^2) - D \sin(\beta t^2)\bigg)^2,$$
then $\tilde{C} \xrightarrow{a.s.} 0$ and $\tilde{D} \xrightarrow{a.s.} 0.$ 
\end{lemma}
\begin{proof}
The proof of this lemma follows along the same lines as Lemma \ref{over_estimation_sin_lses}.\\
\end{proof}
\justify
Now we provide the proof of the fact that the sequential LSEs have the same asymptotic distribution as the LSEs.
\justify
\textit{Proof of Theorem \ref{asymptotic_distribution_first_component}:}
 \textit{(a)} By Taylor series expansion of $\boldsymbol{Q}'_1(\tilde{\boldsymbol{\theta}}_1^{(1)})$ around the point ${\boldsymbol{\theta}_1^0}^{(1)}$, we have:
$$(\tilde{\boldsymbol{\theta}}_1^{(1)} - {\boldsymbol{\theta}_1^0}^{(1)}) = -\boldsymbol{Q}'_1({\boldsymbol{\theta}_1^0}^{(1)}) [\boldsymbol{Q}''_1(\bar{\boldsymbol{\theta}}_1^{(1)})]^{-1}$$
Multiplying both sides by the matrix $\boldsymbol{D}_1^{-1}$, where $\boldsymbol{D}_1 = diag(\frac{1}{\sqrt{n}}, \frac{1}{\sqrt{n}},\frac{1}{n\sqrt{n}})$, we get:
$$(\tilde{\boldsymbol{\theta}}_1^{(1)} - {\boldsymbol{\theta}_1^0}^{(1)})\boldsymbol{D}_1^{-1} = -\boldsymbol{Q}'_1({\boldsymbol{\theta}_1^0}^{(1)})\boldsymbol{D}_1 [\boldsymbol{D}_1\boldsymbol{Q}''_1(\bar{\boldsymbol{\theta}}_1^{(1)})\boldsymbol{D}_1]^{-1}$$
First we show that $\boldsymbol{Q}'_1({\boldsymbol{\theta}_1^0}^{(1)})\boldsymbol{D}_1 \rightarrow N_3(0, 4 \sigma^2 c \boldsymbol{\Sigma}_1^{(1)}).$\\
To prove this, we compute the elements of the derivative vector $\boldsymbol{Q}'_1({\boldsymbol{\theta}_1^0}^{(1)})$:
\begin{flalign*}
& \frac{\partial Q_1({\boldsymbol{\theta}_1^0}^{(1)})}{\partial A_1} = -2 \sum_{t=1}^n \bigg(\sum_{j=2}^p(A_j^0 \cos(\alpha_j^0 t) + B_j^0 \sin(\alpha_j^0 t))  + \sum_{k=1}^q (C_k^0 \cos(\beta_k^0 t^2) + D_k^0 \sin(\beta_k^0 t^2)) + X(t)\bigg)\cos(\alpha_1^0 t),&\\
& \frac{\partial Q_1({\boldsymbol{\theta}_1^0}^{(1)})}{\partial B_1} = -2 \sum_{t=1}^n \bigg(\sum_{j=2}^p(A_j^0 \cos(\alpha_j^0 t) + B_j^0 \sin(\alpha_j^0 t))  + \sum_{k=1}^q (C_k^0 \cos(\beta_k^0 t^2) + D_k^0 \sin(\beta_k^0 t^2)) + X(t)\bigg)\sin(\alpha_1^0 t),&\\
& \frac{\partial Q_1({\boldsymbol{\theta}_1^0}^{(1)})}{\partial \alpha_1} = -2 \sum_{t=1}^n t \bigg(\sum_{j=2}^p(A_j^0 \cos(\alpha_j^0 t) + B_j^0 \sin(\alpha_j^0 t))  + \sum_{k=1}^q (C_k^0 \cos(\beta_k^0 t^2) + D_k^0 \sin(\beta_k^0 t^2)) + X(t)\bigg)\times & \\
& \qquad \qquad \qquad \bigg(-A_1^0 \sin(\alpha_1^0 t) + B_1^0 \cos(\alpha_1^0 t)\bigg).&
\end{flalign*}
Using Lemma \ref{preliminary_result_5}, it can be shown that:
\begin{equation*}
\boldsymbol{Q}'_1({\boldsymbol{\theta}_1^0}^{(1)})\boldsymbol{D}_1 \overset{a.eq.}{=} -2 \begin{pmatrix}
\frac{1}{\sqrt{n}}\sum\limits_{t=1}^{n} X(t) \cos(\alpha_1^0 t) \\
\frac{1}{\sqrt{n}}\sum\limits_{t=1}^{n} X(t) \sin(\alpha_1^0 t) \\
\frac{1}{n\sqrt{n}}\sum\limits_{t=1}^{n} t X(t)(-A_1^0\sin(\alpha_1^0 t) + B_1^0 \cos(\alpha_1^0 t)) 
\end{pmatrix}.
\end{equation*}
Now using CLT, we have:
$$\boldsymbol{Q}'_1({\boldsymbol{\theta}_1^0}^{(1)})\boldsymbol{D}_1 \rightarrow N_3(0, 4 \sigma^2 c \boldsymbol{\Sigma}_1^{(1)})$$ 
Next, we compute the elements of the second derivative matrix, $\boldsymbol{D}_1 \boldsymbol{Q}''_1({\boldsymbol{\theta}_1^0}^{(1)})\boldsymbol{D}_1$. 
By straightforward calculations and using lemmas \ref{preliminary_result_1}, \ref{preliminary_result_2}, \ref{preliminary_result_3} and \ref{preliminary_result_4}, it is easy to show that:
\begin{equation*}
\boldsymbol{D}_1 \boldsymbol{Q}''_1({\boldsymbol{\theta}_1^0}^{(1)})\boldsymbol{D}_1 = 2\boldsymbol{\Sigma}_1^{(1)}.
\end{equation*}
Thus, we have the desired result. \\
\justify
\textit{(b)} Consider the error sum of squares $Q_2(\boldsymbol{\theta}^{(2)}) = \sum\limits_{t=1}^{n}\bigg(y_1(t) - C\cos(\beta t^2) - D\sin(\beta t^2)\bigg)^2$. Here $y_1(t) = y(t) - \tilde{A} \cos(\tilde{\alpha} t) - \tilde{B} \sin(\tilde{\alpha} t)$, $t = 1, \cdots, n$. Let $\boldsymbol{Q}'_2(\boldsymbol{\theta}^{(2)})$ be the first derivative vector and $\boldsymbol{Q}''_2(\boldsymbol{\theta}^{(2)})$, the second derivative matrix of $Q_2(\boldsymbol{\theta}^{(2)})$. By Taylor series expansion of $\boldsymbol{Q}'_2(\tilde{\boldsymbol{\theta}}_1^{(2)})$ around the point ${\boldsymbol{\theta}_1^0}^{(2)}$, we have:
$$(\tilde{\boldsymbol{\theta}}_1^{(2)} - {\boldsymbol{\theta}_1^0}^{(2)}) = -\boldsymbol{Q}'_2({\boldsymbol{\theta}_1^0}^{(2)}) [\boldsymbol{Q}''_2(\bar{\boldsymbol{\theta}}_1^{(2)})]^{-1}$$
Multiplying both sides by the matrix $\boldsymbol{D}_2^{-1}$, where $\boldsymbol{D}_2 = diag(\frac{1}{\sqrt{n}}, \frac{1}{\sqrt{n}},\frac{1}{n^2\sqrt{n}})$, we get:
$$(\tilde{\boldsymbol{\theta}}_1^{(2)} - {\boldsymbol{\theta}_1^0}^{(2)})\boldsymbol{D}_2^{-1} = -\boldsymbol{Q}'_2({\boldsymbol{\theta}_1^0}^{(2)})\boldsymbol{D}_2 [\boldsymbol{D}_2\boldsymbol{Q}''_2(\bar{\boldsymbol{\theta}}_1^{(2)})\boldsymbol{D}_2]^{-1}$$
Now using \eqref{relation_sinusoid_comp_and_estimate}, and proceeding exactly as in part \textit{(a)}, we get:
$$(\tilde{\boldsymbol{\theta}}_1^{(2)} -{\boldsymbol{\theta}_1^0}^{(2)})\boldsymbol{D}_2^{-1} \xrightarrow{d} N_3(0, \sigma^2 c {\boldsymbol{\Sigma}_1^{(2)}}^{-1}).$$
Hence, the result.\\ \qed
\end{appendices}


\begin{thebibliography}{30}
\bibitem{1986}  
Abatzoglou, T. J., 1986 \href{http://ieeexplore.ieee.org/stamp/stamp.jsp?arnumber=4104290}{"Fast maximnurm likelihood joint estimation of frequency and frequency rate."} IEEE Transactions on Aerospace and Electronic Systems, 6, pp. 708-715.


\bibitem{1990_1}
Djuric, P. M., and Kay, S. M., 1990 \href{http://ieeexplore.ieee.org/stamp/stamp.jsp?arnumber=61538} {"Parameter estimation of chirp signals."} IEEE Transactions on Acoustics, Speech, and Signal Processing, 38(12), pp. 2118-2126.


\bibitem{2009}
Fuller, W.A., 2009. Introduction to statistical time series (Vol. 428). John Wiley \& Sons.

\bibitem{2018}
Grover, R., Kundu, D. and Mitra, A., 2018 \href{https://arxiv.org/pdf/1804.01269.pdf} {"On approximate least squares estimators of parameters on one-dimensional chirp signal."} Statistics, (to appear).


\bibitem{1997}
Ikram, M. Z., Abed-Meraim, K. and Hua, Y., 1997 \href{http://ac.els-cdn.com/S1051200497902864/1-s2.0-S1051200497902864-main.pdf?_tid=4430bc46-6ec2-11e7-ad19-00000aacb35f&acdnat=1500716810_80c53c229d30f17ef6e4d4812811e175}{ "Fast quadratic phase transform for estimating the parameters of multicomponent chirp signals."} Digital Signal Processing, 7(2), pp. 127-135.



\bibitem{2008_1}
Kundu, D. and Nandi, S., 2008 \href{http://home.iitk.ac.in/~kundu/paper117.pdf} {"Parameter estimation of chirp signals in presence of stationary noise."} Statistica Sinica, pp. 187-201.

\bibitem{2012}
Kundu, D. and Nandi, S., 2012 \href{http://www.springer.com/in/book/9788132206279} {Statistical Signal Processing: Frequency Estimation.} New Delhi.

\bibitem{2011_2}
Lahiri, A. 2011. Estimators of Parameters of Chirp Signals and Their Properties. PhD thesis, Indian Institute of Technology, Kanpur.

\bibitem{2013}
Lahiri, A., Kundu, D. and Mitra, A., 2013 \href{https://link.springer.com/article/10.1007/s13571-012-0048-x}{ "Efficient algorithm for estimating the parameters of two dimensional chirp signal."},  Sankhya B, 75(1), pp. 65-89.

\bibitem{2014}
Lahiri, A., Kundu, D. and Mitra, A., 2014 \href{http://home.iitk.ac.in/~kundu/chirp-one-LAD-rev-2.pdf} {"On least absolute deviation estimators for one-dimensional chirp model."} Statistics, 48(2), pp. 405-420.

\bibitem{2015}
Lahiri, A., Kundu, D. and Mitra, A., 2015 \href{http://www.sciencedirect.com/science/article/pii/S0047259X15000329} {"Estimating the parameters of multiple chirp signals."} Journal of Multivariate Analysis, 139, pp. 189-206.



\bibitem{2016}
Mazumder, S., 2017 \href{http://www.tandfonline.com/doi/pdf/10.1080/03610918.2015.1053921?needAccess=true} {"Single-step and multiple-step forecasting in one-dimensional single chirp signal using MCMC-based Bayesian analysis."} Communications in Statistics-Simulation and Computation, 46(4), pp. 2529-2547.

\bibitem{2004}
Nandi, S. and Kundu, D., 2004 \href{http://home.iitk.ac.in/~kundu/paper91.pdf} {"Asymptotic properties of the least squares estimators of the parameters of the chirp signals."} Annals of the Institute of Statistical Mathematics, 56(3), pp. 529-544.

\bibitem{1991_1}
Peleg, S. and Porat, B., 1991 \href{http://ieeexplore.ieee.org/stamp/stamp.jsp?arnumber=85033}{ "Linear FM signal parameter estimation from discrete-time observations."} IEEE Transactions on Aerospace and Electronic Systems, 27(4), pp. 607-616.

\bibitem{2008_2}
Prasad, A., Kundu, D. and Mitra, A., 2008 \href{https://ac.els-cdn.com/S0378375807002339/1-s2.0-S0378375807002339-main.pdf?_tid=55e7ea63-65bc-4282-bff0-71059d4eb442&acdnat=1524753741_5b9ec71db42a44e37d6a70dbdf04a879} {"Sequential estimation of the sum of sinusoidal model parameters."} Journal of Statistical Planning and Inference, 138(5), pp. 1297-1313.

\bibitem{1988}
Rice, J. A. and Rosenblatt, M., 1988 \href{http://www.jstor.org/stable/pdf/2336597.pdf?refreqid=excelsior\%3Aa43446ea2e53233ab811fcca4e024ea4} { "On frequency estimation."} Biometrika, 75(3), pp. 477-484.

\bibitem{1961}
Richards, F. SG., 1962 \href{https://www.researchgate.net/publication/268494756_A_method_of_Maximum_Likelihood_estimation}{"A method of maximum-likelihood estimation."} Journal of the Royal Statistical Society. Series B (Methodological), pp. 469-475.

\bibitem{2002}
Saha, S. and Kay, S. M., 2002 \href{http://ieeexplore.ieee.org/stamp/stamp.jsp?arnumber=978378} {"Maximum likelihood parameter estimation of superimposed chirps using Monte Carlo importance sampling."} IEEE Transactions on Signal Processing, 50(2), pp. 224-230.

\end{thebibliography}
\end{document}